\documentclass[journal]{IEEEtran}
\usepackage{graphicx}
\usepackage{subfigure}
\usepackage{epsfig, psfrag, amsmath, amssymb, amsfonts, latexsym, eucal, balance, indentfirst, amsthm}
\usepackage{dsfont}
\usepackage{cite}
\usepackage{soul,color,xcolor} 
\usepackage{url}
\usepackage{booktabs}

\usepackage{stfloats}

\newtheorem{thm}{Theorem}
\newtheorem{cor}{Corollary}
\newtheorem{lem}{Lemma}

\newtheorem{defn}{Definition}

\makeatletter

\newcommand{\Rmnum}[1]{\expandafter\@slowromancap\romannumeral #1@}
\makeatother

\graphicspath{{figs/}}

\begin{document}

\sloppy
\title{Cross-link RTS/CTS for MLO mm-Wave WLANs}
\author{Zhuoling Chen, Yi Zhong, \emph{Senior Member, IEEE}, Martin Haenggi, \emph{Fellow, IEEE}
\thanks{
Zhuoling Chen and Yi Zhong are with the School of Electron. Inf. \& Commun., Huazhong University of Science and Technology, Wuhan, China. 
Martin Haenggi is with the Dept. of Electrical Engineering, University of Notre Dame, US. 
The research has been supported by the National Natural Science Foundation of China (NSFC) grant No. 62471193.
Preliminary findings of this work were presented in the 23rd International Symposium on Modeling and Optimization in Mobile, Ad Hoc, and Wireless Networks (WiOpt 2025) \cite{11123194}. 

The corresponding author is Yi Zhong (yzhong@hust.edu.cn).}
}
\maketitle

\thispagestyle{plain} 

\begin{abstract}
The directional RTS/CTS mechanism of mm-wave Wi-Fi hardly resolves the hidden terminal problem perfectly.
This paper proposes cross-link RTS/CTS under multi-link operation (MLO) to address this problem and introduces a novel point process, named the 
generalized RTS/CTS hard-core process (G-HCP), to model the spatial transceiver relationships under the RTS/CTS mechanism, including the directional case and the omnidirectional case.
Analytical expressions are derived for the intensity, the mean interference, an approximation of the success probability, and the expected number of hidden nodes for the directional RTS/CTS mechanism.
Theoretical and numerical results demonstrate the performance difference between two RTS/CTS mechanisms.
The cross-link RTS/CTS mechanism ensures higher link quality at the cost of reduced network throughput.
In contrast, the directional RTS/CTS sacrifices the link quality for higher throughput.
Our study reveals a fundamental trade-off between link reliability and network throughput, providing critical insights into the selection and optimization of RTS/CTS mechanisms in next-generation WLAN standards.
\end{abstract}
\section{Introduction}
\label{sec:intro}
\subsection{Motivations}
Over the past decades, Wireless Local Area Network (WLAN) standards have evolved into one of the most widely adopted and efficient networking technologies. Among them, Wi-Fi has become the dominant solution for providing ubiquitous wireless connectivity. The first Wi-Fi standard, released by IEEE in 1997, enabled a maximum throughput of 2 Mbps \cite{620533}, thereby formalizing WLAN technology in the 2.4 GHz band. However, such performance has quickly become inadequate to satisfy the rapidly growing demand for high-speed wireless access.

To meet these demands, subsequent Wi-Fi generations have continuously enhanced their physical (PHY) and medium access control (MAC) capabilities. Notable advances include the introduction of Orthogonal Frequency-Division Multiplexing (OFDM) in IEEE 802.11a (Wi-Fi 2) \cite{815305} and Orthogonal Frequency-Division Multiple Access (OFDMA) in IEEE 802.11ax (Wi-Fi 6), as well as the extension of usable spectrum resources \cite{9793689}. In particular, Wi-Fi 6E further expanded into the 6 GHz band \cite{9165719}, while IEEE 802.11ad and 802.11ay pioneered the use of millimeter-wave (mm-wave) communications. 
Despite these advances, mm-wave transmissions suffer from severe penetration losses and blockage sensitivity, which fundamentally limit their large-scale deployment.

Along with the large-scale adoption of Wi-Fi \cite{ericsson2023mobility}, medium access coordination has emerged as a critical challenge. The hidden terminal problem, in particular, significantly degrades network performance. To mitigate this issue, IEEE introduced the Request-to-Send/Clear-to-Send (RTS/CTS) handshake \cite{1092767}. In mm-wave systems, beamforming is widely employed to combat high-frequency propagation impairments \cite{7306370}. However, the resulting directional transmissions render the conventional RTS/CTS mechanism less effective, as neighboring nodes outside the beam path may fail to detect control frames, thereby causing collisions (see Fig.~\ref{fig:hiddennode}).

\begin{figure*}
\centering
\includegraphics[width=0.75\textwidth]{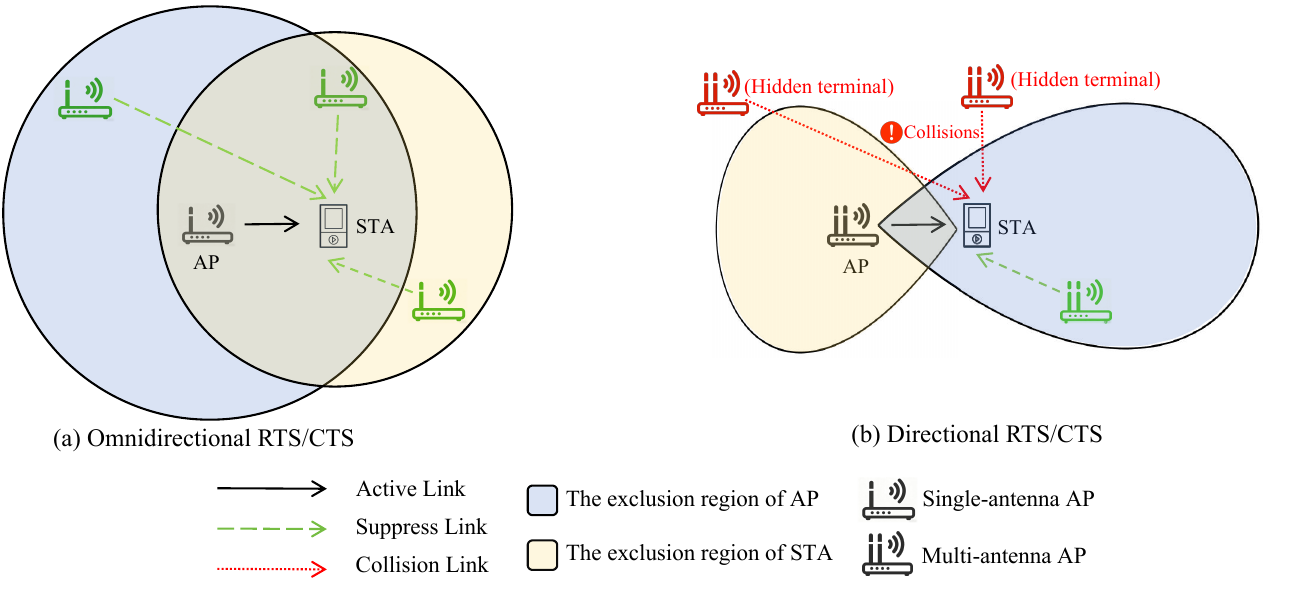}
\caption{
Illustration of the hidden terminal problem under omnidirectional and directional RTS/CTS mechanisms.
Three types of APs are shown: (i) black APs that are actively communicating with their associated Stations (STAs), (ii) green APs that are silenced by the RTS/CTS mechanism and thus prohibited from transmitting, and (iii) the remaining APs that act as hidden terminals and may cause collisions.
In the omnidirectional case, nearby APs detect the RTS/CTS frames and defer transmissions. 
In the directional case, the beamformed transmission prevents AP1 and AP2 from overhearing the RTS/CTS frames, even though they are located near the STA. As a result, AP1 and AP2 may initiate transmissions during the ongoing communication, leading to collisions.}
\label{fig:hiddennode}
\end{figure*}


To address these limitations, the IEEE 802.11be standard (Wi-Fi 7) has recently introduced multi-link operation (MLO), which allows simultaneous use of both sub-7 GHz and mm-wave links \cite{10058126,9855457,10978791}. This integration provides new opportunities to enhance MAC-layer coordination in mm-wave Wi-Fi. Motivated by this, we propose a novel \textit{cross-link RTS/CTS mechanism} that leverages the sub-7 GHz band for control frame exchange while maintaining mm-wave transmissions for high-rate data (see Fig.~\ref{fig:mld}). Specifically, RTS and CTS frames are transmitted in the sub-7 GHz band so that neighboring devices, regardless of beam orientation, can correctly update their Network Allocation Vector (NAV) and avoid collisions in the mm-wave band.

This work is thus motivated by the necessity to systematically model and analyze the spatial characteristics of both the conventional directional RTS/CTS and the proposed cross-link RTS/CTS mechanisms. By providing rigorous performance benchmarks, we aim to guide the design and optimization of next-generation WLAN protocols that integrate sub-7 GHz and mm-wave technologies under the MLO framework.

\subsection{Related Work}
The MLO framework in Wi-Fi 7 (IEEE 802.11be) represents a fundamental architectural milestone for future mm-wave WLANs. By integrating sub-7 GHz and mm-wave technologies, MLO enables next-generation wireless applications with higher throughput, lower latency, and enhanced reliability \cite{IEEE802.11_IMMW_2024}. Building upon this vision, the IMMW Study Group has proposed a Project Authorization Request (PAR) \cite{par}, which defines a new standard, IEEE 802.11bq, aimed at extending the reuse of sub-7 GHz PHY/MAC specifications to the mm-wave band under the 802.11be MLO framework. Despite these advances, the design of a robust RTS/CTS handshake mechanism tailored for integrated sub-7 GHz and mm-wave systems remains largely unexplored in the current literature.

In parallel, extensive research has been devoted to characterizing mm-wave networks through both empirical measurements and theoretical models \cite{6932503,8550813,9494282,68340753,6736994,7913628}. For instance, \cite{8550813} employs stochastic geometry to analyze mm-wave networks within finite regions using a tractable flat-top antenna pattern. While analytically convenient, the flat-top model introduces inaccuracies in evaluating network success probability. To mitigate this issue, \cite{7913628} proposes sinc and cosine antenna patterns, which significantly improve approximation accuracy and enable more precise derivations of network success probability. However, these models are predominantly developed under the assumption of a Poisson point process (PPP), which is poorly aligned with networks governed by RTS/CTS protocols.

To address this discrepancy, traditional hard-core point processes, such as the Matérn hard-core model \cite{chen2024characterizing,5934671,4215725}, have been adopted to study CSMA-based networks.
\cite{4215725} derive closed-form analytical expressions for crucial network performance metrics based on the type \Rmnum{2} Matérn hard-core model.
Yet, these models fall short of capturing the unique characteristics of RTS/CTS-enabled networks, where both transmitters and receivers enforce exclusion regions. To overcome this limitation, \cite{11018845} introduces the dual-zone hard-core process (DZHCP), which explicitly incorporates dual exclusion regions around transmitters and receivers. Although promising, the DZHCP framework is currently limited to omnidirectional antenna scenarios and cannot accurately model directional RTS/CTS mechanisms in mm-wave networks.

Another line of research focuses on the analytical tractability of the success probability in hard-core point processes, which is intractable in closed form. The work in \cite{11192202} derives an analytical bound under the densely packed assumption. To overcome this limitation, several studies have developed the Approximate SIR Analysis based on the PPP (ASAPPP) framework \cite{6897962,7322270,8648502}, which enables tractable approximations while preserving key spatial correlations. Initially developed for cellular models, ASAPPP provides accurate approximations by mapping hard-core processes onto PPP-based equivalents. More recently, \cite{11018845} extended the applicability of ASAPPP to bipolar models and demonstrated its validity in this broader context. Nonetheless, most existing studies employing ASAPPP or related methods remain focused on sub-7 GHz networks, typically assuming Rayleigh fading, which limits their applicability to directional mm-wave scenarios with RTS/CTS mechanisms.

\begin{figure}
    \centering
    \includegraphics[width=0.48\textwidth]{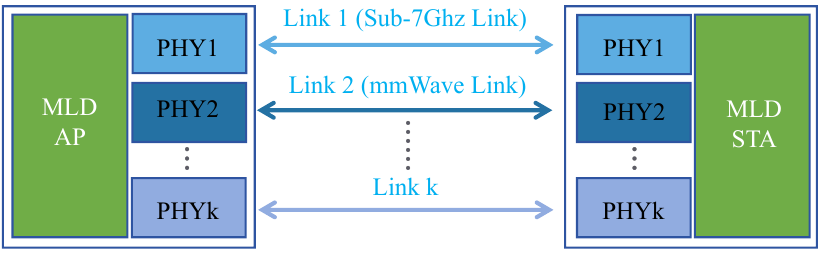}
    \caption{ Illustration of the mm-wave Wi-Fi supporting MLO.
    This mm-wave Wi-Fi devices are equipped with multiple independent physical layer for each link, enabling simultaneous communication on different frequency bands or channels.
    The connection between AP and STA includes multiple links including at least one link in the Sub-7 GHz band and another in the mm-wave band.
    If there is any blockage in the mm-wave link, the Sub-7 GHz link can ensure communication quality.}
    \label{fig:mld}
\end{figure}

\subsection{Contributions}
This paper focuses on the two types of RTS/CTS handshake mechanism.
We investigate the mean interference and the approximate success probability in mm-wave networks with the blockage effect reflected by a LOS ball blockage model under the RTS/CTS handshake mechanism.
The key contributions of our study are summarized as follows:

\begin{itemize}
    \item We propose a novel exclusion region model based on cosine antenna patterns, applicable to both directional and omnidirectional scenarios. The formulation maintains analytical tractability in performance analysis, specifically for the directional case.
    \item Based on the exclusion region, we introduce a novel hard-core point process, named generalized RTS/CTS hard-core process (G-HCP), to characterize the spatial properties of directional RTS/CTS mechanism and the cross-link RTS/CTS mechanism. Leveraging classical statistical measures and analytical tools, we derive the expressions for mean interference.
    \item We extend the ASAPPP method to approximate the success probability under integer-parameter Nakagami-$M$ fading channels, accounting for a finite LOS communication region. An exact expression for the mean interference-to-signal ratio (MISR) in a bounded domain is established to support this approximation.
    \item In the context of directional RTS/CTS, we provide an exact expression for the expected number of hidden nodes, revealing the impact of system parameters such as antenna array size, frequency, and node density.
\end{itemize}


\section{Stochastic Model for RTS/CTS-Enabled Networks}
\label{sec:model}
\subsection{Network and Channel Assumptions}

We consider a MLD network where each transceiver can operate simultaneously in the Sub-7 GHz and mm-wave bands. All devices employ a single omnidirectional antenna for Sub-7 GHz communications, while the mm-wave transmission relies on uniform linear arrays (ULAs) with $N_{\rm t}$ (for MLD transmitters) and $N_{\rm r}$ (for MLD receivers) antenna elements. The mm-wave reception is omnidirectional.
Omnidirectional transmission is restricted to the Sub-7 GHz band and is triggered only when (i) the cross-link RTS/CTS mechanism is enabled, and (ii) the transmitted frame is a cross-link RTS or CTS frame. For all other cases, mm-wave directional transmission is applied, whereas reception remains omnidirectional.

Blockage is modeled using the LOS ball model \cite{6932503}, with LOS radius $R$ representing the maximum distance within which potential LOS interferers exist. Given that NLOS paths are typically over 20 dB weaker than LOS paths \cite{6834753}, we restrict our analysis to LOS channels. The mm-wave propagation follows a clustered channel model, where the channel vector between transmitter $y$ and its receiver is given by
\begin{equation}
\mathbf{h}_{y}=\sqrt{N_{\rm t}}\rho_{y}\mathbf{a}_\mathrm{t}^H(\vartheta_{y}),
\end{equation}
with fading coefficient $\rho_{y}$ modeled as Nakagami-$M$, and $\mathbf{a}_\mathrm{t}(\vartheta_{y})$ is the array response vector as
\begin{equation}
\mathbf{a}_{\mathrm{t}}(\vartheta_{y}) = \frac{1}{\sqrt{N_{\mathrm{t}}}} \left[1, \cdots, e^{j 2\pi k \vartheta_{y}}, \cdots, e^{j 2\pi (N_{\mathrm{t}}-1) \vartheta_{y}}\right]^T,
\end{equation}
where \( \vartheta_y = \frac{d_0}{\lambda} \sin \phi_y \), and \( d_0 \), \( \lambda \), and \( \phi_y \) represent the antenna spacing, wavelength, and physical angle of departure (AoD).

The optimal analog beamforming vector between transmitter $y$ and its intended receiver is 
\begin{equation}
\mathbf{w}_y = \mathbf{a}_{\mathrm{t}}(\varphi_{y}),
\end{equation}
yielding the beamforming gain
\begin{equation}
G_{\text{act}}(\vartheta_{y}-\varphi_{y}) = \left| \mathbf{a}_{\mathrm{t}}^H(\vartheta_{y}) \mathbf{w}_y \right|^2 
= \frac{\sin^2(\pi N_{\mathrm{t}}(\vartheta_{y}-\varphi_{y}))}{N_{\mathrm{t}}^2 \sin^2(\pi (\vartheta_{y}-\varphi_{y}))}.
\end{equation}
For tractability, we approximate this gain by a cosine antenna pattern \cite{7913628} in (\ref{Gt}).



\subsection{Generalized RTS/CTS Hard-Core Process}
We model the locations of transceivers using a Poisson bipolar process, where the transmitters are distributed according to a homogeneous PPP, each paired with a receiver at a fixed distance $d$ and a uniformly random angle $\theta$. The RTS/CTS mechanism is incorporated through a thinning procedure, which selectively removes transceiver pairs from the underlying process based on the RTS/CTS handshake rules. This allows us to capture both the directional RTS/CTS and the cross-link RTS/CTS in a unified framework.

The RTS/CTS operation induces exclusion regions around each transceiver pair, which consist of an RTS-cleaned region and a CTS-cleaned region. For omnidirectional RTS/CTS, the exclusion regions are isotropic, whereas directional RTS/CTS generates anisotropic exclusion regions due to beamforming.

We distinguish between two types of RTS/CTS thinning:

- \emph{Type \Rmnum{1} thinning}: A transceiver pair is retained only if no other potential transmitter lies within its exclusion region. This corresponds to a synchronous channel access model where time is slotted and only one pair can be active within a given exclusion region. Collisions occur if multiple transmitters attempt to access the channel simultaneously.

- \emph{Type \Rmnum{2} thinning}: Each transceiver is assigned a random time mark representing its channel access attempt. A transceiver is retained if no other transmitter with a smaller time mark lies within its exclusion region. This model reflects asynchronous contention, where multiple transceivers may coexist in a slot but priority is given to those with earlier access marks.

Both thinning schemes ensure that retained transceivers respect the RTS/CTS exclusion principle, with Type \Rmnum{1} emphasizing strict mutual exclusion and Type \Rmnum{2} offering a more realistic representation of distributed channel access.

To characterize the antenna effect, we approximate the actual beam pattern. 
For simplicity, we assume that the AoD between each transmitter and its associated receiver is zero, so the optimal analog beamforming vector reduces to
\begin{equation}
\label{4}
\mathbf{w}_{y} = \frac{1}{\sqrt{N_{\mathrm{t}}}} \left[1,1,\ldots,1\right]^T.
\end{equation}
The corresponding beamforming gain is
\begin{equation}
G_{\text{act}}(\vartheta_{y}) = \left| \mathbf{a}_{\mathrm{t}}^H(\vartheta_{y}) \mathbf{w}_y \right|^2 
= \frac{\sin^2(\pi N_{\mathrm{t}} \vartheta_{y})}{N_{\mathrm{t}}^2 \sin^2(\pi \vartheta_{y})}.
\end{equation}

Following \cite{7913628}, we approximate this beamforming gain by the cosine model
\begin{equation}
\label{5}
G_{\cos}(\vartheta) =
\begin{cases}
\cos^2\!\left(\tfrac{\pi N_{\mathrm{t}}}{2}\vartheta \right), & |\vartheta| \leq \tfrac{1}{N_{\mathrm{t}}}, \\
0, & \text{otherwise},
\end{cases}
\end{equation}
where $\vartheta = \tfrac{d_0}{\lambda}\sin \phi$.  

Using the small-angle approximation $\sin(x) \approx x$ for small $x$, we further obtain a simplified directional gain
\begin{equation}
\label{Gt}
G_{\rm t}(\phi) =
\begin{cases}
\cos^{2}\!\left(\tfrac{\pi N_{\mathrm{t}} d_0}{2 \lambda} \phi \right), & |\phi| \leq \tfrac{\lambda}{d_0 N_{\mathrm{t}}}, \\
0, & \text{otherwise}.
\end{cases}
\end{equation}

\begin{figure}
    \centering
    \includegraphics[width=0.45\textwidth]{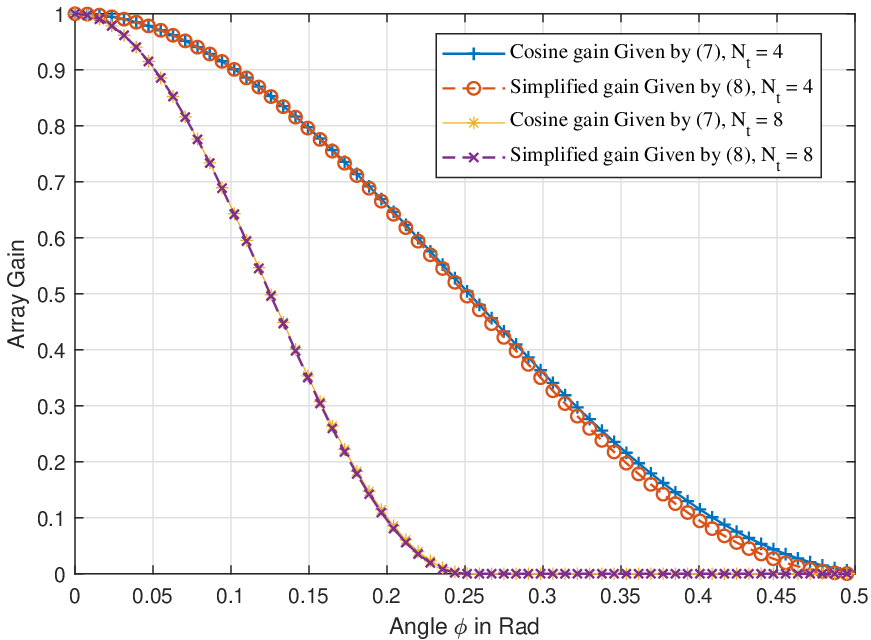}
    \caption{ The comparisons between the different gains.    }
    \label{fig:pattern}
\end{figure}

In Fig. \ref{fig:pattern}, we observe that the simplified gain provides a good approximation for the cosine gain.
Based on the simplified gain, we define the RTS/CTS exclusion regions.  
For transmitter $y$ and its receiver $x$, the RTS exclusion region is
\begin{equation}\begin{aligned}
S_{\rm t}(y) \triangleq \bigg\{ z \in \mathbb{R}^2 \mid \|\overrightarrow{yz}\| \leq R_{\rm t} \sqrt{G_{\rm t}(\phi)}, \;
\cos\phi = \frac{\overrightarrow{yz}\cdot \overrightarrow{yx}}{\|\overrightarrow{yz}\| \|\overrightarrow{yx}\|} \bigg\},
\end{aligned}\end{equation}
while the CTS exclusion region is
\begin{equation}
S_{\rm r}(x) \triangleq \bigg\{ z \in \mathbb{R}^2 \mid \|\overrightarrow{xz}\| \leq R_{\rm r} \sqrt{G_{\rm r}(\phi)}, \;
\cos\phi = \frac{\overrightarrow{xz}\cdot \overrightarrow{xy}}{\|\overrightarrow{xz}\| \|\overrightarrow{xy}\|} \bigg\},
\end{equation}
with the receiver gain
\begin{equation}
G_{\rm r}(\phi) =
\begin{cases}
\cos^{2}\!\left(\tfrac{\pi N_{\mathrm{r}} d_0}{2 \lambda} \phi \right), & |\phi| \leq \tfrac{\lambda}{d_0 N_{\mathrm{r}}}, \\
0, & \text{otherwise}.
\end{cases}
\end{equation}
Here, $R_{\rm r}$ and $R_{\rm r}$ denote the maximum transmission ranges of RTS and CTS frames.
Therefore, the exclusion region of the transceiver pair with transmitter $y$ and receiver $x$ is 
\begin{equation}
S(y,x)=S_{\rm t}(y) \cup S_{\rm r}(x).
\end{equation}

When the antenna spacing is 0, $d_0=0$, all antennas can be regarded as a single antenna, which corresponds to the omnidirectional case.
And the numbers of antennas ($N_{\rm t}$ and $N_{\rm r}$) represent the transmission powers.
In the special case, both $S_{\rm t}(y,x)$ and $S_{\rm r}(y,x)$ reduce to disks centered at $y$ and $x$. 
Therefore, the proposed exclusion region framework generalizes naturally to both omnidirectional (cross-link RTS/CTS) and directional RTS/CTS mechanisms, enabling a unified spatial model.

\begin{figure}
    \centering
    \includegraphics[width=0.48\textwidth]{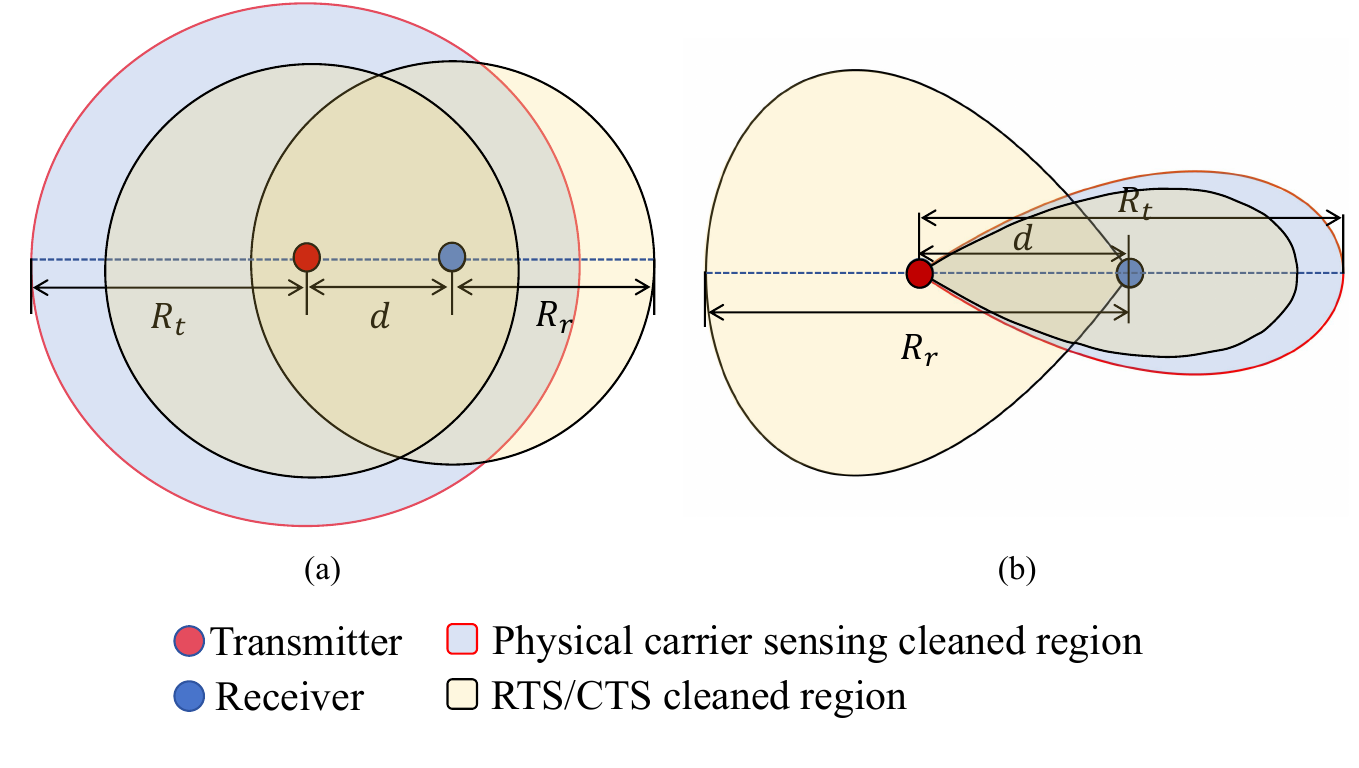}
    \caption{ Illustration of the exclusion region. 
    (a) is the case of omnidirectional case, i.e., the cross-link RTS/CTS mechanism.
    (b) is the case of directional case with $N_{\rm r}$=16 and $N_{\rm r}$=8.
    The exclusion region consists of two parts, i.e. the physical carrier sensing cleaned region and the RTS/CTS cleaned region.
    }
    \label{fig:exclusionregion}
\end{figure}

\subsection{Mathematical Description}
The G-HCP model employs a dependently marked PPP, denoted as ${\widetilde\Phi}_{\rm PPP}$, for spatial distribution of the transmitters.
This process is structured as follows:
\begin{itemize}
    \item $\Phi_{\rm PPP} = \{y_i\} \subset \mathbb{R}^2$ denotes the locations of potential transmitters modeled by a PPP with intensity $\lambda_p$;
    \item $x_i$ denotes the location of the receiver. $\theta_i$ is an i.i.d. orientation of the receiver $x_i$ for the transmitter $y_i$, uniformly distributed in $[0, 2\pi]$. Having assumed a constant distance of $d$ between a transmitter and its receiver, the orientation $\theta_i$ together with $y_i$ uniquely determines the receiver $x_i$.
    These marks determine the time order of transmissions for each transmitter;
    \item $e_i$ is the medium access indicator.
    Each transceiver pair $(y_i, x_i)$ has a medium access indicator, which determines whether the transmitter could transmit.
\end{itemize}

\subsubsection{\emph{Type \Rmnum{1} G-HCP}}
For the Type \Rmnum{1} G-HCP, let $\mathcal{N}(y, x,e)$ denotes the other transmitters within the exclusion region of transceiver pair $(y, x)$:
\begin{align}
\mathcal{N}&(y,x,e)\triangleq\big\{(y',x',e')\in\widetilde{\Phi}_{\rm PPP}\setminus\{(y,x,e)\}: \nonumber\\
&y'\in S_{\rm t}(y)\cup S_{\rm r}(x)\big\}.
\end{align}
The medium access indicator $e_i$ is a dependent mark as 
\begin{equation}
e_i=\mathds{1}(\#\mathcal{N}(y_i,x_i,e_{i})=0),
\end{equation}
where $\#$ denotes the count of elements in the neighborhood.

The set of transceiver pairs that have passed all RTS/CTS contention requirements, is denoted by $\tilde{\Phi}$
\begin{equation}
\tilde\Phi\triangleq \{(y,x,e)\in\tilde\Phi_{\rm PPP}\colon e=1\}.
\end{equation}

Finally, the set of active transmitters is defined as $\Phi$
\begin{equation}
 \Phi\triangleq \{(y,x,e)\in\tilde\Phi\colon y\}
\end{equation}

\subsubsection{\emph{Type \Rmnum{2} G-HCP}}
The difference between the Type \Rmnum{1} and the Type \Rmnum{2} is the time mark $t$, which reflects the time cost of a transmitter competing for channel access in the Type \Rmnum{2}.
\begin{itemize}
    \item $t_i$ is an i.i.d. time mark uniformly distributed in $[0, 1]$, representing the time cost of a transmitter competing for channel access, i.e., the relative initiation time of transmission for each transmitter.
\end{itemize}

For the Type \Rmnum{2} G-HCP, let $\mathcal{N}(y,x,t,e)$ be the neighborhood of the transmitters, which have smaller time marks than $t$ and are within the exclusion region of transceiver pair $(y,x)$.
\begin{align}
\mathcal{N}(y,x,t,e)\triangleq&\big\{(y',x',t',e'):(y',x',e')\in\widetilde{\Phi}_{\rm PPP}\setminus\{(y,x,e)\},
\nonumber\\
&y'\in S_{\rm t}(y)\cup S_{\rm r}(x), t'<t\big\}.
\end{align}
The medium access indicator $e_i$ of $(y_i,x_i,t_i,e_i)$ is a dependent mark defined as 
\begin{equation}
e_i=\mathds{1}(\#\mathcal{N}(y_i,x_i,t_i,e_i)=0).
\end{equation}

The set of transceiver pairs that have passed all RTS/CTS contention requirements, is denoted by \begin{equation}
\tilde\Phi\triangleq \{(y,x,e,t)\in\tilde\Phi_{\rm PPP}\colon e=1\}.
\end{equation}

Finally, the set of active transmitters is defined as $\Phi$
\begin{equation}
 \Phi\triangleq \{(y,x,e,t)\in\tilde\Phi\colon y\}
\end{equation}

The G-HCP is used to describe the spatial distribution of the active transmitters within a time slot.
$\Phi_{\rm PPP}$ is the set of the locations of the transmitters that attempt to transmit within this time slot, which means $\Phi_{\rm PPP}$ is reallocated within each slot.
For the Type \Rmnum{2} G-HCP, the time marks are redrawn within each slot. Consequently, in both Type \Rmnum{1} G-HCP and Type \Rmnum{2} G-HCP, a transmitter's suppression is a temporary state per slot, not a permanent condition.

\section{Network Performance Analysis}
\label{sec:i}
In this section, we analyze the network performance based on the proposed G-HCP model. Without loss of generality, we consider the typical transmitter $y_o$ located at the origin $o$ and its associated receiver $x_o$ positioned at $(d, 0)$.
To facilitate the derivations, we define several key quantities as follows
\begin{itemize}
    \item Transceiver pair parameters $(r,\beta,\theta)$: each transceiver pair is represented by $(r,\beta,\theta)$, where the transmitter and receiver are located at $y=(r \cos \beta,  r \sin \beta)$ and $x= (r \cos \beta + d \cos \theta,  r \sin \beta + d \sin \theta)$, respectively.

    \item Union exclusion region $V(r,\beta,\theta)$: the spatial region covered by the union of the exclusion zones of the typical transceiver and a transceiver at \((r, \beta, \theta)\).
    \item Spatial criteria for interference $S_1$, $S_2$, $S_3$ and $S_4$: Define the spatial criteria for transmitters within the exclusion regions, helping to determine interference conditions:
    \begin{itemize}
        \item $S_1 = \{(r, \beta, \theta) : r^2 \leq R_{\mathrm{t}}^2G_{\rm t}(\theta-\beta-\pi)\}$,
        \item $S_2 = \{(r, \beta, \theta) : r^2 + 2 r d \cos(\beta - \theta) + d^2 \leq R_{\mathrm{r}}^2G_{\rm r}(\arccos\frac{-r\cos(\theta-\beta)-d}{\sqrt{r^2 + 2 r d \cos(\beta - \theta) + d^2}})\}$,
        \item $S_3 = \{(r, \beta, \theta) : r^2 \leq R_{\mathrm{t}}^2G_{\rm t}(\beta)\}$,
        \item $S_4 = \{(r, \beta, \theta) : r^2 - 2 r d \cos\beta + d^2 \leq R_{\mathrm{r}}^2G_{\rm r}(\arccos\frac{r\cos\beta-d}{\sqrt{r^2 - 2 r d \cos\beta + d^2}})\}$,
    \end{itemize}
\end{itemize}

The set $S_1$ includes $(r, \beta, \theta)$ where $o$ is within its transmitter's exclusion region.
The set $S_2$ includes $(r, \beta, \theta)$ where $o$ is within its receiver's exclusion region.
The set $S_3$ includes $(r, \beta, \theta)$ of which the transmitter is within $S_{\rm t}(o)$.
The set $S_4$ includes $(r, \beta, \theta)$ of which the transmitter is within $S_{\rm r}(x_o)$.


\subsection{Node Intensity}
In the G-HCP model, the point process $\Phi$ is obtained through a dependent thinning of the original PPP $\Phi_{\rm PPP}$. Under the Palm probability measure $\mathbb{P}^{\Phi_{\rm PPP}}$, the node intensity of the retained process $\Phi_{\rm PPP}$ can be expressed as
\begin{equation}
    \lambda_b = \lambda_p \, \mathbb{P}^{\Phi_{\rm PPP}}(\text{$o$ retained}),
\end{equation}
where $\mathbb{P}^{\Phi_{\rm PPP}}(\text{node retained})$ denotes the mean retention probability of a node in $\Phi_{\rm PPP}$ after applying the thinning rule.

\subsubsection{Type \Rmnum{1} G-HCP}
In the Type \Rmnum{1} G-HCP, a node is retained only if its exclusion region is free of any other nodes from $\Phi_a$. This implies that each retained node corresponds to an empty exclusion zone. Consequently, the retention probability equals the void probability of a PPP with intensity $\lambda_p$, i.e.,
\begin{equation}
    \mathbb{P}^{\Phi_{\rm PPP}}(\text{$o$ retained}) = \exp(-\lambda_p V_o),
\end{equation}
where $V_o$ is the area of the exclusion region associated with the typical node. 
Hence, the resulting node intensity is
\begin{equation}
    \lambda_b = \lambda_p e^{-\lambda_p V_o}. \label{equ:lambda_1}
\end{equation}

When the parent intensity $\lambda_p$ is small, the retained intensity $\lambda_b$ increases approximately linearly with $\lambda_p$. However, as $\lambda_p$ becomes large, the exponential term dominates, leading to a decrease in $\lambda_b$. Therefore, $\lambda_b$ reaches its maximum value at $\lambda_p = 1/V_o$, beyond which it monotonically decreases due to the strong exponential decay of $\mathbb{P}^{\Phi_{\rm PPP}}(\text{$o$ retained})$.

\subsubsection{Type \Rmnum{2} G-HCP}
In the Type \Rmnum{2} G-HCP, each transceiver pair is associated with a time mark $t \in [0,1]$, drawn independently from a uniform distribution. The time mark determines the priority of a node during the contention process. Specifically, a transmitter $y$ with time mark $t$ is retained only if no other transmitters with smaller time marks are located within its exclusion region. 
The overall node intensity is obtained as
\begin{equation}
    \lambda_b = \lambda_p \int_0^1 e^{-\lambda_p t V_o} \, \mathrm{d}t = \frac{1}{V_o}\big(1 - e^{-\lambda_p V_o}\big). \label{equ:lambda_2}
\end{equation}


It can be seen from \eqref{equ:lambda_2} that $\lambda_b$ is a monotonically increasing function of $\lambda_p$, asymptotically approaching its upper limit of $1/V_o$ as $\lambda_p \to \infty$. 

\subsection{Mean Interference}
In this part, we will derive the mean interference at the typical transceiver pair in G-HCP.

The aggregate interference at the receiver $x_o$ is
\begin{equation}
\mathrm{I}_{x_{o}}=\sum_{y\in\Phi\setminus \{o\}}P_0N_\mathrm{t}|\rho_{y}|^2 G_{\rm t} \, l(\|y-x_{o}\|),
\end{equation}
where $P_0$ is the per-antenna transmit power, the channel gain $|\rho_{y}|^2$ follows a gamma distribution $\mathrm{Gamma}(M,\frac{1}{M})$ and $l(\cdot)$ is the path-loss function.

Firstly, we derive the mean interference in the network, where all interferers are LOS interferers.
The mean interference for the dual zone hard-core point process has been established in \cite{11018845}. 
Building upon this result, we extend the derivation to obtain the mean interference for the directional RTS/CTS modeled by G-HCP.

\begin{thm}
The mean interference $I_{x_0}$ experienced by the typical receiver $x_0$ in G-HCP, under the assumption that all interferers are LOS interferers, is
\begin{equation}\begin{aligned}
\label{tI1}
 & \mathbb{E}_{(o, 0)}^{!}\left[I_{x_o}\right]=\frac{\lambda_{p}^{2}P_{0}N_{\rm t}}{2\pi\lambda_b}\int_{0}^{\infty}\int_{0}^{2\pi}\int_{0}^{2\pi} \\
 & l\left(\sqrt{r^{2}-2rd\cos{\beta}+d^{2}}\right)G_{\rm t}(r,\beta,\theta)k(r,\beta,\theta)r\mathrm{d}\theta\mathrm{d}\beta\mathrm{d}r,
\end{aligned}\end{equation}
where $k(r,\beta,\theta)$ for the Type \Rmnum{1} process is
\begin{equation}
\label{k1}
k(r,\beta,\theta)=
\begin{cases}
0 \quad \text{if} (r,\beta,\theta)\in (S_1\cup S_2)\cup(S_3\cup S_4), \\
\exp\left(-\lambda_p V(r,\beta,\theta)\right) \quad \text{otherwise,} 
\end{cases}\end{equation}
and $k(r,\beta,\theta)$ for the Type \Rmnum{2} process is
\begin{equation}
\label{k2}
k(r,\beta,\theta)=
\begin{cases}
0 & (r,\beta,\theta)\in (S_1\cup S_2)\cap(S_3\cup S_4), \\
2p(V) & (r,\beta,\theta)\in\overline{S_1}\cap\overline{S_2}\cap\overline{S_3}\cap\overline{S_4}, \\
p(V) & \text{otherwise,}
\end{cases}\end{equation}
 in which V is the abbreviation of $V(r,\beta,\theta)$, and $p(V)$ is
 \begin{equation}p(V)=\frac{V_oe^{-\lambda_pV}-Ve^{-\lambda_pV_o}+V-V_o}{\lambda_p^2(V-V_o)VV_o}.\end{equation}
\end{thm}
\begin{proof}
The key distinction between our model and that in \cite{11018845} lies in considering the beamforming gain.
Due to the typical receiver $x_o$ positioned at $(d, 0)$, the gain can be written as a function of $(r,\beta,\theta)$, i.e., $G_{\rm r}(r,\beta,\theta)$, where $(r,\beta,\theta)$ is the coordinate of the interfering transmitter.
In addition, the product of the path loss and the beamforming gain can be treated as an equivalent path loss. 
This allows us to extend the analytical approach presented in Theorems 1 and 2 of \cite{11018845} to complete the proof.
\end{proof}

Theorem 1 assumes that all interferers are LOS interferers.
Actually, in our model, mm-wave signals fall into two types: LOS and NLOS, which means that this assumption is impractical.
In this paper, we neglect the impact of NLOS signals due to their significantly weaker channel gains compared to LOS paths.
Under this simplification, the mean interference in the LOS ball model can be derived using Theorem 1.
\begin{cor}
    The mean interference $I_{x_o,R}$ experienced by the typical receiver $x_0$ in G-HCP with a LOS radius R is
    \begin{equation}\begin{aligned}
\label{tI2}
 &\mathbb{E}_{(o, 0)}^{!}\left[I_{x_o,R}\right]=\frac{\lambda_p^{2}P_{0}N_{\rm t}}{2\pi\lambda_b}\int_{0}^{2\pi}\int_{0}^{2\pi}\int_{0}^{\sqrt{R^2-(d\sin\beta)^2}+d\cos\beta}\\
&l\left(\sqrt{r^{2}-2rd\mathrm{cos}\beta+d^{2}}\right)G_{\rm t}(r,\beta,\theta)k(r,\beta,\theta)r\mathrm{d}\beta\mathrm{d}\theta\mathrm{d}r,
\end{aligned}
\end{equation}
where $k(r,\beta,\theta)$ is given by (\ref{k1}) for Type \Rmnum{1} and by (\ref{k2}) for Type \Rmnum{2}.
\end{cor}

Compared to Theorem 1, Corollary 1 calculates the mean interference within the LOS region, which is a disk area of radius $R$ centered at the receiver.
As $R\to \infty$, Corollary 1 becomes equivalent to Theorem 1.
Physically, this describes a scenario that all links satisfy the LOS condition, which is consistent with the assumption of Theorem 1.

\subsection{Success Probability}
\label{sec:p}

\begin{figure*}
\begin{minipage}[t]{0.45\textwidth}
\centering
    \includegraphics[width=\columnwidth]{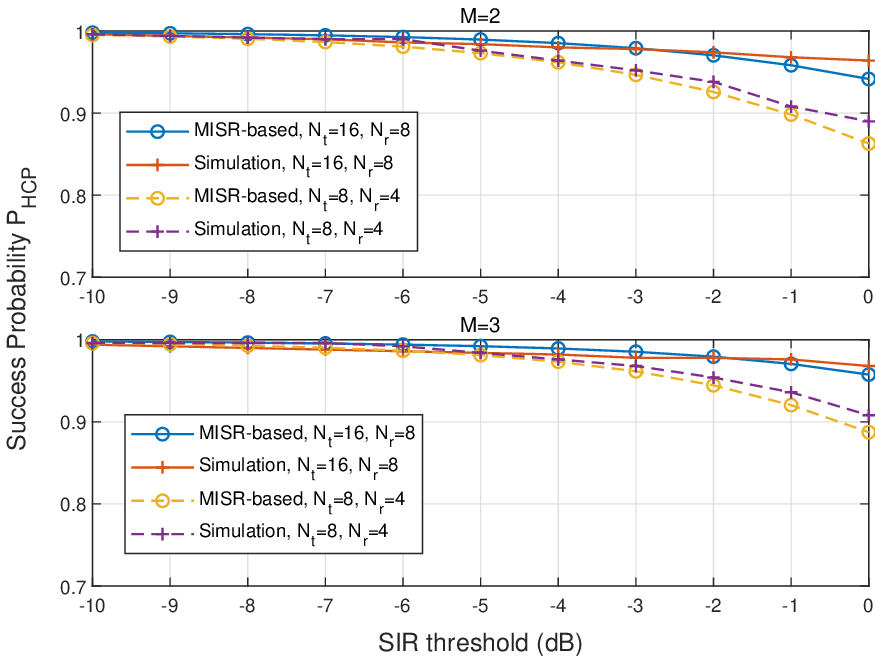}
    \caption{SIR ccdf for the type \Rmnum{1} hard-core process and MISR-based approximation for \(\lambda_p=4\times10^{-4} {\rm m}^{-2}\) and $R=300$m under the directional RTS/CTS mechanism.}
    \label{fig:ccdftype1}
    \end{minipage}
\hfill
    \begin{minipage}[t]{0.45\textwidth}
    \centering
    \includegraphics[width=\columnwidth]{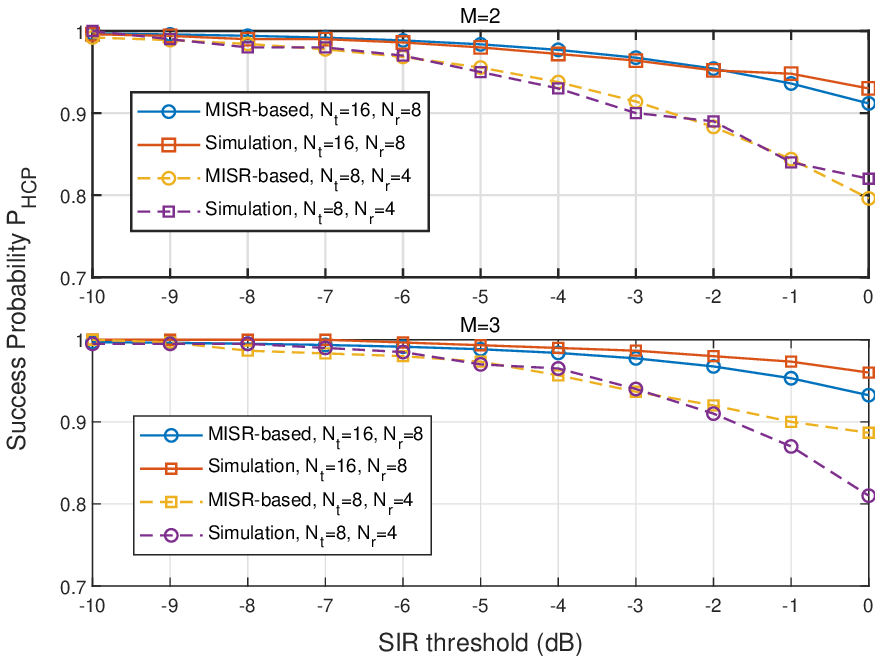}
    \caption{SIR ccdf for the type \Rmnum{2} hard-core process and MISR-based approximation for \(\lambda_p=4\times10^{-4} {\rm m}^{-2}\) and R=300m under the directional RTS/CTS mechanism.}
    \label{fig:ccdftype2}
    \end{minipage}
\end{figure*}
The transmission success probability serves as a fundamental performance metric in wireless network analysis.
However, for the hard-core process under consideration, deriving an exact analytical expression for the transmission success probability appears intractable. 

To address this issue, we employ ASAPPP to approximate the success probability.
This approach is motivated by prior work \cite{7913628}, which derived the success probability $P_{\rm PPP}(\theta)$ for a PPP model with nearest transmitter association under the LOS ball assumption.

In this study, we extend the applicability of ASAPPP to Nakagami-$M$ fading, where $M$ is the integer, for the bipolar network.
The success probability $P_{\phi}(\theta)$ for a point process $\phi$ is approximated via $P_{\rm PPP}(\theta)$ adjusted by $G$
\begin{align}
    P_{\phi}(\theta) \approx P_{\rm PPP}\bigg(\frac{\theta}{G}\bigg),
\end{align}
where the asymptotic gain $G$ defined in \cite{6897962} and reflects the relationship between the SIR characteristics of the PPP and $\phi$.

We will derive the asymptotic gain $G$ for both Type I and Type II G-HCP under the LOS ball assumption.


\begin{lem}
    The asymptotic gain for both Type I and Type II G-HCP with a LOS radius R
    \begin{align}
        G&=\nonumber
        \mathrm{MISR}_R\frac{2\pi\lambda_bl(d)}{\lambda_p^{2}P_{0}N_{\rm t}}
        \bigg(\int_{0}^{2\pi}\int_{0}^{2\pi}\int_{0}^{\sqrt{R^2-(d\sin\beta)^2}+d\cos\beta}\\
&l\left(\sqrt{r^{2}-2rd\mathrm{cos}\beta+d^{2}}\right)G_{\rm t}(r,\beta,\theta)k(r,\beta,\theta)r\mathrm{d}\beta\mathrm{d}\theta\mathrm{d}r\bigg)^{-1},
    \end{align}
    where $k(r,\beta,\theta)$ is given by (\ref{k1}) for Type \Rmnum{1} and is given by (\ref{k2}) for Type \Rmnum{2} and $\mathrm{MISR}_{R}$ is the mean interference-to-average-signal ratio of a PPP model with nearest transmitter association and a LOS radium R, written as
    \begin{align}
        \mathrm{MISR}_R&= \nonumber
        \frac{2}{\alpha -2}-\frac{\alpha}{\alpha -2} {}F_1(\lambda_p\pi R^2)\\& -\frac{4\lambda_p\pi R^2}{\alpha^2 -4}F_{2}(\lambda_p\pi R^2)+e^{-\lambda_p\pi R^2},
    \end{align}
    with $F_n(x)$ denoting the confluent hypergeometric function
    \begin{align}
        F_n(x)={}_{1}F_{1}\bigg(\frac{\alpha}{2};n+\frac{\alpha}{2};-x\bigg).
    \end{align}
\end{lem}
\begin{proof}
    See Appendix \ref{appendix:D}.
\end{proof}

\begin{figure*}
\begin{equation}
\label{ck}
c_k(r)=
\begin{cases}
\pi \lambda_p \left\{r^{2}-R^{2}+\delta R^{2}\mathbb{E}_{g}\big[\mathrm{E}_{1+\delta}(sR^{-\alpha}g)\big]-\delta r^{2}\mathbb{E}_{g}\big[\mathrm{E}_{1+\delta}(sr^{-\alpha}g)\big]\right\}, & k=0,\\
\frac{\pi\delta\lambda_ps^k}{k!}\left\{R^{2-\alpha k}\mathbb{E}_g\big[g^k\mathrm{E}_{1+\delta-k}(sR^{-\alpha}g)\big]-r^{2-\alpha k}\mathbb{E}_g\big[g^k\mathrm{E}_{1+\delta-k}(sr^{-\alpha}g)\big]\right\}, & k\geq1.
\end{cases}
\end{equation}
\end{figure*}

\begin{thm}
\label{thm:succ_prob}
The success probability of a receiver in Type \Rmnum{1} and Type \Rmnum{2} G-HCP, under LOS conditions with a LOS radius R and independent Nakagami-$M$ fading with integer M, can be approximated as
\begin{align}
\nonumber
P_{\rm HCP}(\theta)\approx
2\pi\lambda_p\int_{0}^{R} e^{-\pi \lambda_p r ^2}\|\exp{\{\mathbf{C}_{M}(r)\}}\|_{1}r\mathrm{d}r,
\end{align}
where $\|{\cdot}\|_{1}$ represents the sum of the first column and
\begin{align}
\nonumber
\mathbf{C}_{M}(r)=\left[
\begin{array}
{ccccc}{c_{0}(r)} & & & & \\
c_{1}(r) & c_{0}(r) & & & \\
c_{2}(r) & c_{1}(r) & c_{0}(r) & & \\
\vdots & & & \ddots & \\
c_{M-1}(r) & \cdots & c_{2}(r) & c_{1}(r) & c_{0}(r)
\end{array}\right],
\end{align}
whose exponent is given by (\ref{ck}), where $s=\frac{M\theta r^{\alpha}}{G}$, $\delta=2/\alpha$, $G$ is defined in Lemma 1 and $E_{p}(\cdot)$ is the generalized exponential integral.
\end{thm}

\begin{proof}
    See Appendix \ref{appendix:B}.
\end{proof}
To assess the accuracy of the MISR approximation for the Nakagami-$M$ fading, we compare its results with simulation data in this work. 
Figures \ref{fig:ccdftype1} and \ref{fig:ccdftype2} present the success probability under varying values of $M$ and for commonly used values of SIR threshold under the directional RTS/CTS mechanism.
The proposed MISR approximation provides accurate approximations for practical success probabilities at low SIR thresholds and the gap between theoretical and simulated results expands with increasing SIR threshold.

The success probability for cross-link RTS/CTS is consistently observed to be approximately 1 in simulations and exhibits negligible dependence on the SIR threshold.
Due to the directional transmission and the isotropic exclusion region, the asymptotic gain $G$ is greater than 350. This high gain leads to an analytical approximation of 1, which is in excellent agreement with the simulation results.
Considering this almost invariant behavior, the results are omitted from Figure \ref{fig:ccdftype1} and Figure \ref{fig:ccdftype2}.

\subsection{Hidden Node}
\label{sec:hd}
Our preceding analysis demonstrates that the directional RTS/CTS mechanism cannot completely eliminate the hidden terminal problem. 
We now present a quantitative analysis of this limitation. Specifically, this section develops: an exact mathematical derivation of the expected number of hidden nodes under directional RTS/CTS operation.
We focus on the hidden nodes of the typical receiver $x_o$.
\begin{defn}
\label{defn1}
A \emph{hidden node} $y$ of $x_o$ in the G-HCP is defined as a transmitter that satisfies the following two conditions:
\begin{enumerate}
    \item $y \in \Phi$, i.e., the transmitter is active (retained after the thinning process);
    \item $x_o \in S_{\rm t}(y)$, i.e., the typical receiver $x_o$ lies within the exclusion region of transmitter $y$.
\end{enumerate}
In other words, a hidden node is an active transmitter whose exclusion region overlaps with the typical receiver, potentially causing interference despite the RTS/CTS mechanism.
\end{defn}

The exclusion region of each transceiver pair is caused by the RTS/CTS frame.
The other transceiver pairs in the exclusion region of a transceiver pair can decode the RTS/CTS frame transmitted by the transceiver pair.
This formal definition captures the essential characteristics of hidden nodes.

Let $\Phi_{\text{h}}$ be the set of all hidden nodes
\begin{equation}
\Phi_{\rm h}\triangleq\{y:y\in \Phi \setminus\{o\},x_o \in S_{\rm t}(y)\}.
\end{equation}
The number of hidden nodes $N_{\rm h}$ in the G-HCP is denoted as
\[ N_{\rm h} = \#\Phi_{\rm h}, \].

\begin{thm}
\label{hnthm}
The expected number of hidden nodes $N_{\rm h}$ in G-HCP is
\begin{equation}\begin{aligned}
  \mathbb{E}_{(o, 0)}^{!}\left[N_{\rm h}\right]&=\frac{\lambda_p^{2}}{\pi\lambda_b}\int_{0}^{\pi}\int_{0}^{R_{\rm r}\sqrt{1-(\frac{d\sin\beta}{R_{\rm t}})^2}-d\cos\beta}\int_{-f(r_y)}^{f(r_y)}\\
 & k\bigg(r,\beta,\theta+\arcsin\Big(\frac{r\cos\beta-d}{r_y}\Big)\bigg)r\mathrm{d}\beta\mathrm{d}r\mathrm{d}\theta,
\end{aligned}\end{equation}
where $k(r,\beta,\theta)$ is given by (\ref{k1}) for Type \Rmnum{1} and by (\ref{k2}) for Type \Rmnum{2}, $f(x)=\frac{2\lambda}{\pi N_{\rm t} d_0}\arccos(\frac{x}{R_{\rm t}})$ and $r_y=\sqrt{r^2+d^2-2rd\cos\beta}$. 
\end{thm}
\begin{proof}
See Appendix \ref{appendix:C}.
\end{proof}

\section{Implementation of cross-link RTS/CTS Protocol}
\label{sec:implement}
\subsection{Channel Access Procedure}

\begin{figure*}
\centering
\includegraphics[width=1\textwidth]{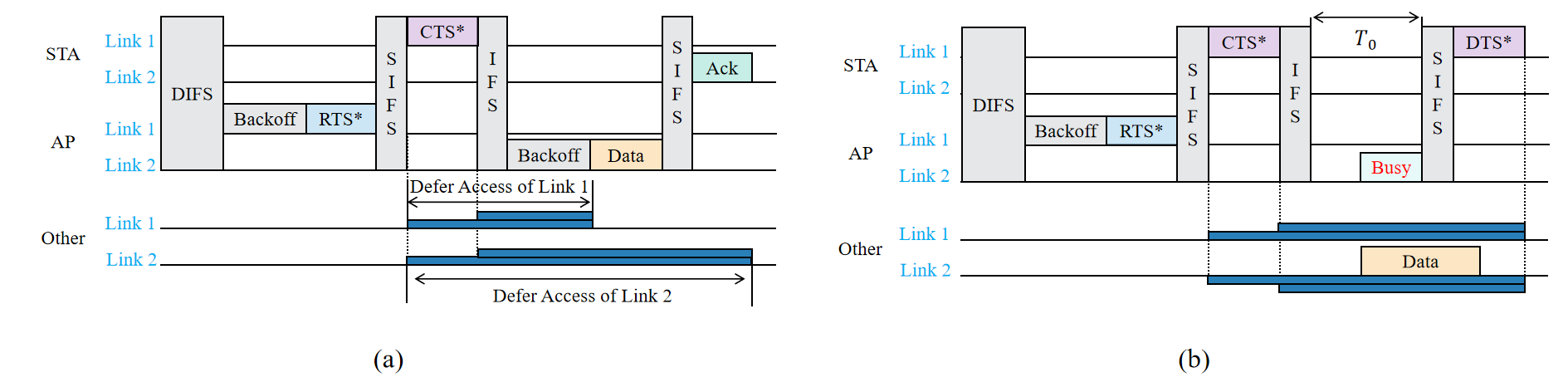}
\caption{
(a) Message exchange process of the cross-link RTS/CTS mechanism: the AP transmits an RTS* (cross-link RTS frame) over Link~1 (Sub-7\,GHz); the target STA responds with a CTS* (cross-link CTS frame) on the same link. Neighboring devices then set their NAV timers on Link~1 (Sub-7\,GHz) according to Duration~1 in the RTS* and CTS* frames, and simultaneously on Link~2 (mm-wave) according to Duration~2. Data transmission subsequently proceeds on the high-rate mm-wave link.  
(b) Message exchange process in the case where the AP fails to occupy the mm-wave channel within the specified time interval.
}
\label{process}
\end{figure*}

Unlike the traditional RTS/CTS mechanism, which operates entirely within a single frequency band, the cross-link RTS/CTS mechanism employs the Sub-7\,GHz band for control signaling while enabling data transmission over the mm-wave band.  
This dual-band coordination improves spatial reuse but also requires simultaneous reservation of resources in both Sub-7\,GHz and mm-wave channels, which constitutes the primary challenge of this mechanism.

The overall channel access procedure, illustrated in Fig.~\ref{process}(a), can be summarized as follows:
\begin{itemize}
    \item {Sub-7\,GHz channel contention:} The AP first contends for access to the Sub-7\,GHz channel;
    \item {cross-link RTS transmission:} Upon successful contention, the AP transmits a cross-link RTS frame over the Sub-7\,GHz link;
    \item {mm-wave channel assessment:} The target STA, after receiving the RTS frame, checks the availability of the mm-wave channel;
    \item {cross-link CTS response:} If the mm-wave channel is idle, the STA responds with a cross-link CTS frame on the Sub-7\,GHz link;
    \item {network reservation:} Neighboring devices overhearing either the RTS or CTS frame set their NAV timers for both the Sub-7\,GHz and mm-wave bands;
    \item {mm-wave Data transmission:} Following a specified Inter-Frame Space (IFS), the AP proceeds with data transmission on the mm-wave channel.
\end{itemize}

To reduce the risk of collisions during the transition to mm-wave transmission, a backoff procedure is introduced before the AP occupies the mm-wave channel. However, excessively long backoff periods would reduce spectral efficiency.  
To balance reliability and efficiency, a timeout mechanism is adopted: if the STA does not receive any data frame from the AP within a predetermined duration $T_0$ after sending the CTS frame, it broadcasts a cross-link DTS frame, thereby instructing neighboring devices to reset their NAV states, as illustrated in Fig.~\ref{process}(b).

\subsection{Control Frame Design}
\begin{figure*}
\centering
\includegraphics[width=0.75\textwidth]{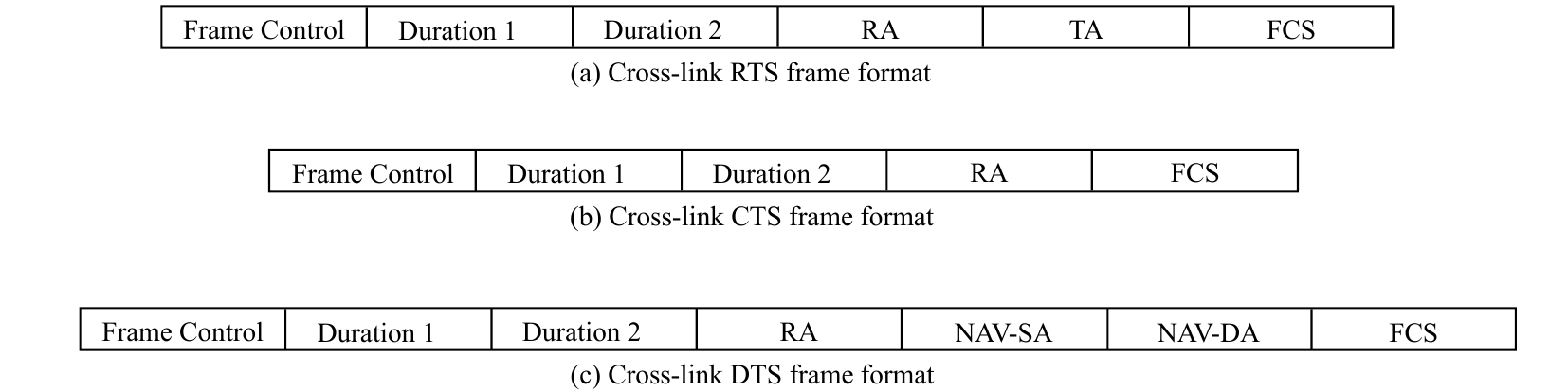}
\caption{Frame formats of the cross-link RTS/CTS mechanism.}
\label{frame}
\end{figure*}

The design of the cross-link RTS/CTS frame formats must accommodate the dual-band reservation property of the mechanism. Specifically, each control frame includes two duration fields, enabling neighboring devices to update their NAV timers on both the Sub-7\,GHz and mm-wave links, as shown in Fig.~\ref{frame}.  

In addition, the standard Frame Check Sequence (FCS) field, a 32-bit cyclic redundancy check (CRC) located at the end of each Wi-Fi frame, is used for error detection.  
Definitions of all other fields are summarized in Table~\ref{tab:frame_fields}.

\begin{table*}[htbp]
  \centering
  \caption{cross-link Control Frame Field Definitions}
  \label{tab:frame_fields}
  \begin{tabular}{p{0.18\textwidth} p{0.2\textwidth} p{0.2\textwidth} p{0.25\textwidth}}
    \toprule
    \textbf{Frame Type} & \textbf{Duration 1 / Duration 2} & \textbf{RA Field} &\textbf{Additional Fields} \\
    \midrule
    \textbf{cross-link RTS} 
    & the time that the entire transmission session occupies the Sub-7 GHz / mm-wave channels 
    & MAC address of the recipient STA 
    & \textbf{TA:} MAC address of the transmitting STA \\
    \midrule
    \textbf{cross-link CTS} 
    & Remaining time on Sub-7 GHz / mm-wave channels
    & TA from the cross-link RTS frame 
    & -- \\
    \midrule
    \textbf{cross-link DTS} 
    & Remaining time on Sub-7 GHz / mm-wave channels 
    & TA from the cross-link RTS frame 
    & \textbf{NAV-SA/NAV-DA:} MAC addresses of the STAs transmitting the cross-link RTS / cross-link CTS \\
    \bottomrule
  \end{tabular}
\end{table*}

\section{Performance Comparison Under a Unified Benchmark}
\label{sec:num}
\subsection{Unified Benchmark}
\label{sec:relationship}
To ensure a fair and consistent comparison between the proposed cross-link RTS/CTS mechanism and the conventional directional RTS/CTS mechanism, it is essential to establish a unified performance benchmark. A key consideration in this comparison is whether to normalize the exclusion region size across different configurations.

Assuming omnidirectional reception, the receive antenna gain is set to unity. A receiver can successfully decode a frame in either the Sub-7~GHz or mm-wave band when the received power satisfies
\begin{equation}
P_{\mathrm{rx}} \geq P_{\mathrm{rx(th)}},
\end{equation}
where $P_{\mathrm{rx(th)}}$ denotes the minimum detectable power threshold for the corresponding frequency band.

The maximum transmission ranges $R_1$ for frequency band 1 and $R_2$ for frequency band 2, respectively, under the Friis model are related by
\begin{equation}
\label{reference}
    \frac{R_{1}}{R_2}=\sqrt{\frac{P_{1}}{P_{2}}\cdot\frac{G_{1}}{G_{2}}}\cdot\frac{f_2}{f_{1}},
\end{equation}
where $P_i$, $G_i$ and $f_i$ are the transmission power of the single antenna, the transmitting antenna gain and frequency of frequency band $i$, respectively. 
For the mm-wave band, the directional antenna gain in the main-lobe direction equals the number of antenna elements, while for the Sub-7~GHz band, the omnidirectional transmit gain is unity.

In this study, the exclusion region size corresponding to the cross-link RTS/CTS mechanism with a transmit power of 20~mW is adopted as the reference.
The exclusion regions for both the directional and cross-link RTS/CTS mechanisms under other transmit powers are then scaled according to (\ref{reference}).

\subsection{Performance Comparison}
\label{sec:nr}
\begin{table}
\centering
\caption{System Parameters}
\label{table:parameter}
\begin{tabular}{c|l|c}
\hline
\textbf{Symbol} & \textbf{Parameter} & \textbf{Value} \\
\hline
$R$ & LOS radius & 300 m \\
$d$ & Transmitter–receiver distance & 20 m \\ 
$N_{\rm r}$ & Number of transmit antenna elements & 16 \\
$N_{\rm r}$ & Number of receive antenna elements & 8 \\
$\theta$ & SIR threshold & $-5$ dB \\
$\alpha$ & Path loss exponent & 2.1 \\
\hline
\end{tabular}
\end{table}

In this section, the terms \textit{cRTS/CTS} and \textit{dRTS/CTS} refer to the cross-link RTS/CTS and directional RTS/CTS mechanisms, respectively. 
For the reference configuration, the receiver range is set to $R_{\rm r} = 4d$ and the transmitter range to $R_{\rm t} = 4.8d$ for the cRTS/CTS scheme operating with a transmission power of 20 mW. 
The default system parameters are summarized in Table~\ref{table:parameter}. 
The transmission power per mm-wave antenna is $P_0 = 20$~mW, and the path loss follows $l(r) = (1 + r^{\alpha})^{-1}$. 

In the following figures, the legend labels “Type~\Rmnum{1}” and “Type~\Rmnum{2}” correspond to the Type~\Rmnum{1} and Type~\Rmnum{2} G-HCPs, respectively. 
The label $P_{\text{sub-7}}$ denotes the sub-7~GHz transmission power for cRTS/CTS. 
Two configurations of cRTS/CTS are evaluated, with $P_{\text{sub-7}} = 20$~mW and $P_{\text{sub-7}} = 40$~mW, respectively.


Figure \ref{fig:lambda} illustrates the relationship between the active node density $\lambda_b$ and the original PPP intensity $\lambda_p$ under various RTS/CTS configurations.
For all Type \Rmnum{1} models, the active node density $\lambda_b$ first increases with $\lambda_p$ and then decreases after reaching a peak.
The maximum occurs when $\lambda_p = \frac{1}{V_o}$, where $V_o$ denotes the area of the exclusion region, yielding a peak value of $\frac{1}{V_o} e^{-1}$.
In contrast, for Type \Rmnum{2} models, $\lambda_b$ increases monotonically with $\lambda_p$ and asymptotically approaches the reciprocal of the exclusion region.
It can also be observed that the dRTS/CTS mechanism yields a higher $\lambda_b$ than cRTS/CTS under identical exclusion conditions, indicating that cRTS/CTS achieves a stronger suppression of concurrent transmissions and thus a more effective interference mitigation.

Figure \ref{fig:interference} depicts the variation of the mean interference with respect to the original PPP intensity $\lambda_p$.
For all Type \Rmnum{1} models, the mean interference first increases with $\lambda_p$ and then decreases after reaching a peak.
A similar trend is observed for Type \Rmnum{2} models, where the mean interference evolves consistently with the active node density.
In particular, the mean interference in Type \Rmnum{2} models increases monotonically with $\lambda_p$ and gradually converges to a constant value. 
It is also observed that dRTS/CTS exhibits higher mean interference than cRTS/CTS under both Type \Rmnum{1} and Type \Rmnum{2} configurations, even though its active node density is lower.
This seemingly counterintuitive behavior stems from the inherent differences in the exclusion regions of the two RTS/CTS mechanisms and the directional gain of antenna arrays.
When $\lambda_p \to 0$, both mechanisms can be approximated as PPPs with identical intensities, since the exclusion effect becomes negligible.
The only distinction lies in the spatial distribution of transmitters: nodes under dRTS/CTS can be located closer to the typical receiver, leading to stronger interference.
Overall, cRTS/CTS provides a clear advantage by maintaining a lower interference level for most density regimes, thereby achieving superior link quality and more stable network performance.

\begin{figure}
    \centering
    \includegraphics[width=0.45\textwidth]{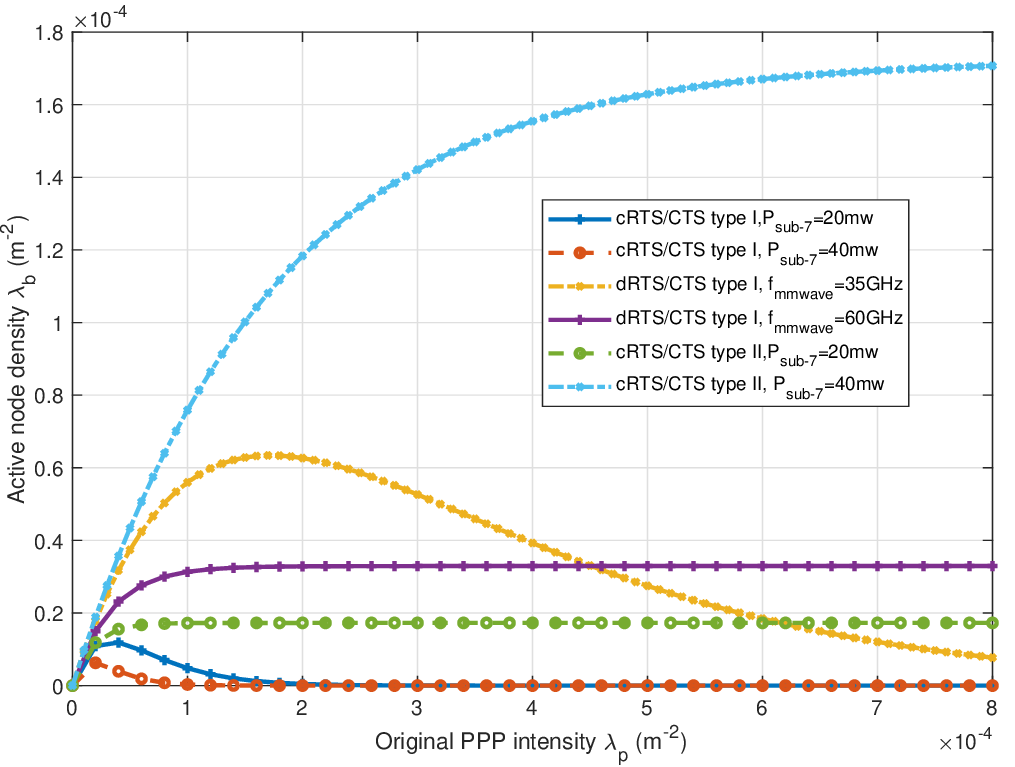}
    \caption{Active node density $\lambda_b$ as a function of the original PPP intensity $\lambda_p$ for various RTS/CTS configurations.}
    \label{fig:lambda}
\end{figure}

\begin{figure}
    \centering
    \includegraphics[width=0.45\textwidth]{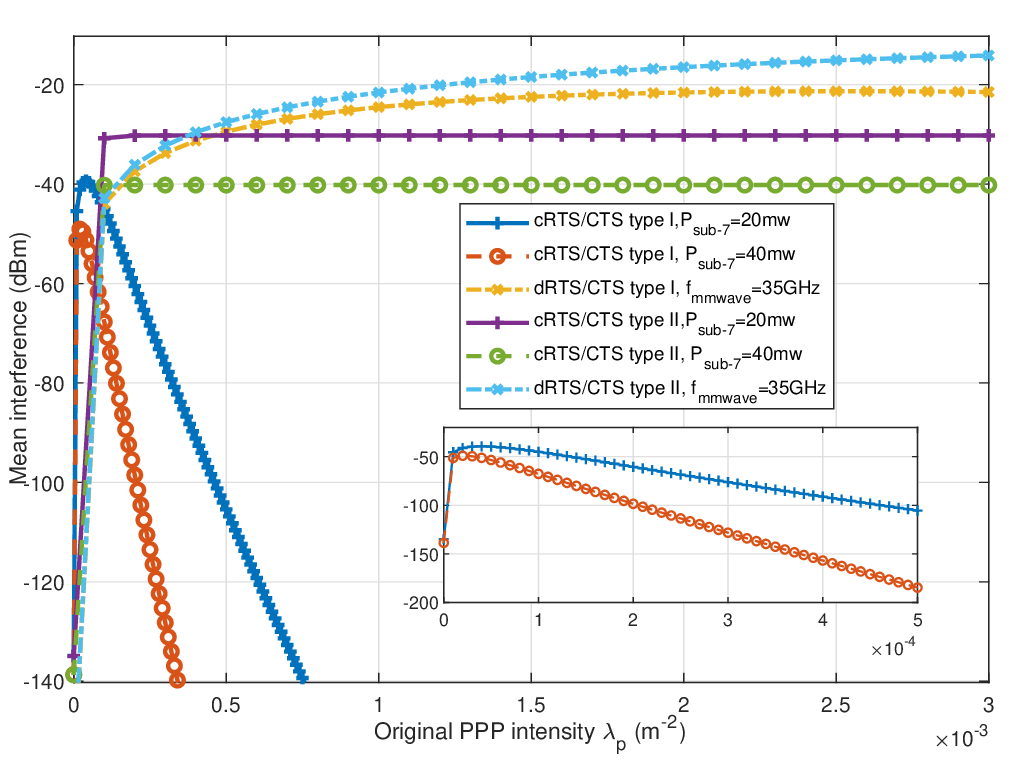}
    \caption{Mean interference as a function of the original PPP intensity $\lambda_p$ under different RTS/CTS mechanism}
    \label{fig:interference}
\end{figure}

Figure \ref{fig:ip} illustrates the variation of a key performance metric, the product of the success probability and the active transmitter density $\lambda_b$.
This metric effectively represents the network throughput.
For Type \Rmnum{1} schemes, the throughput trends closely follow those of the mean interference.
Both cRTS/CTS and dRTS/CTS exhibit an initial increase in throughput as $\lambda_p$ rises, followed by a decline beyond a certain point, indicating the existence of an optimal network density for Type \Rmnum{1} configurations.
A distinct behavior is observed under Type \Rmnum{2} schemes.
For cRTS/CTS, the throughput monotonically follows the variation of mean interference, while for dRTS/CTS, a non-monotonic pattern emerges—throughput first increases, then slightly decreases, and finally converges to a constant level.
This divergence can be explained by the interference characteristics shown in Figure \ref{fig:interference}.
At low $\lambda_p$, throughput for both mechanisms is mainly governed by the active node density.
However, when $\lambda_p > 1\times10^{-3}{\rm m}^{-2}$, the continuously increasing interference in dRTS/CTS leads to a reduction in success probability that outweighs the gain from higher node density, resulting in the observed throughput trend. 
Overall, despite the higher interference levels, dRTS/CTS achieves superior throughput performance compared to cRTS/CTS, benefiting from its more efficient spatial reuse of the wireless medium.

\begin{figure}
    \centering
    \includegraphics[width=0.45\textwidth]{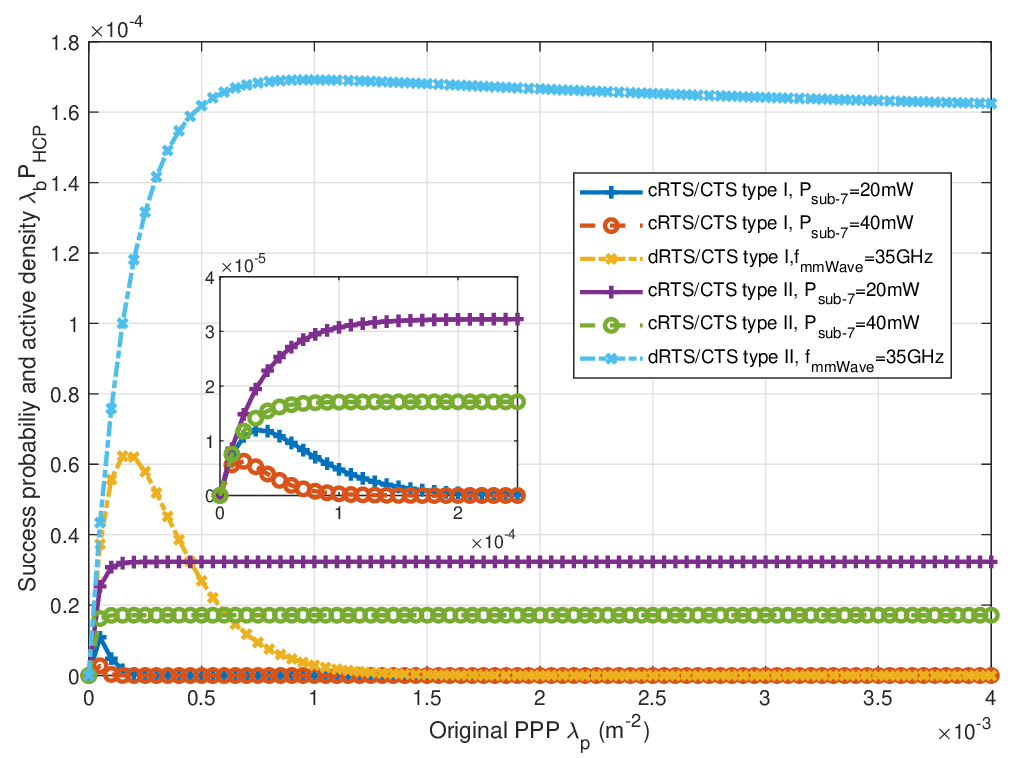}
    \caption{The network throughput as as a function of the original PPP intensity $\lambda_p$. This metric is defined as the product of the success probability and the active transmitter density and provides a measure of network throughput.}
    \label{fig:ip}
\end{figure}

Figure \ref{fig:number1} illustrates the relationship between the mean number of hidden nodes in Type \Rmnum{1} dRTS/CTS and the original PPP intensity $\lambda_p$ under different system configurations.
For Type \Rmnum{1} models, the mean number of hidden nodes first increases with $\lambda_p$ and then decreases after reaching a peak, indicating the existence of an optimal intensity that maximizes hidden-node occurrence.
A comparison between the second curve (35 GHz, $N_{\rm r}=16$, $N_{\rm r}=8$, $d=20$ m) and the fifth curve (60 GHz, $N_{\rm r}=16$, $N_{\rm r}=8$, $d=20$ m) shows that higher carrier frequency results in a larger maximum mean number of hidden nodes.
For a fixed frequency (35 GHz) and link distance ($d=10$ m), increasing $N_{\rm r}$ significantly raises the peak value, whereas increasing $N_{\rm r}$ yields only a marginal impact.
This observation indicates that the hidden-node count is mainly influenced by the transmitter-side exclusion region, since a larger $N_{\rm r}$ allows more distant transmitters to become hidden nodes despite the narrower beam coverage.
Comparing the first (35 GHz, $N_{\rm r}=16$, $N_{\rm r}=8$, $d=10$ m) and second (35 GHz, $N_{\rm r}=16$, $N_{\rm r}=8$, $d=20$ m) curves reveals that both achieve similar maximum values but at different $\lambda_p$, implying that merely scaling the exclusion region size, without altering its spatial structure, is insufficient to mitigate the hidden-node issue.
Overall, all curves exhibit peak values exceeding 3.5, highlighting the severity of the hidden-node problem in Type \Rmnum{1} dRTS/CTS within a moderate density range.
As $\lambda_p$ continues to increase, the mean number of hidden nodes gradually converges to zero, suggesting that the likelihood of hidden-node occurrence diminishes in denser Type \Rmnum{1} networks.

\begin{figure}
    \centering
    \includegraphics[width=0.45\textwidth]{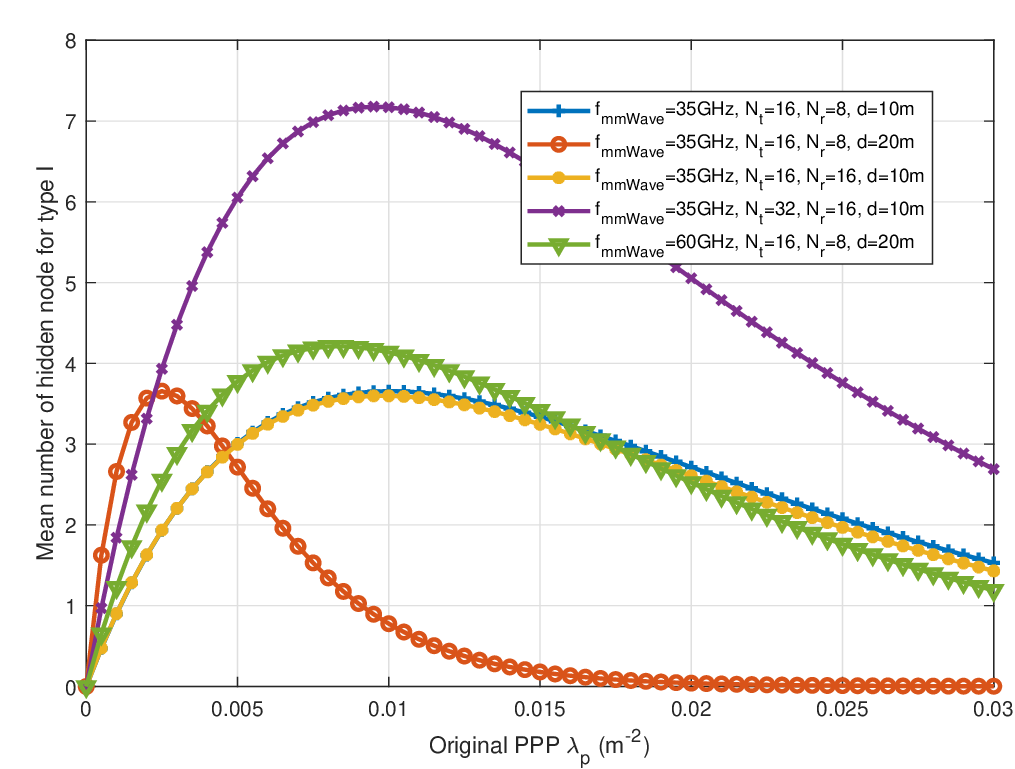}
    \caption{The mean number of hidden node for Type \Rmnum{1} under different system configurations}
    \label{fig:number1}
\end{figure} 

Figure \ref{fig:number2} illustrates the relationship between the mean number of hidden nodes in Type \Rmnum{2} dRTS/CTS and the original PPP intensity $\lambda_p$ under different system configurations.
For Type \Rmnum{2} models, the mean number of hidden nodes increases monotonically with $\lambda_p$ and gradually converges to a steady-state value as $\lambda_p$ grows large.
This steady value exhibits similar parameter-dependent behavior to the peak value observed in the Type \Rmnum{1} case.
The fundamental distinction between Type \Rmnum{1} and Type \Rmnum{2} lies in their asymptotic behavior: in Type \Rmnum{1}, increasing node density eventually suppresses hidden-node occurrences, whereas in Type \Rmnum{2}, the problem persists even at high densities.
Furthermore, both the steady value for Type \Rmnum{2} and the peak value for Type \Rmnum{1} are consistently higher in dRTS/CTS than in cRTS/CTS, indicating a more severe hidden-node issue and consequently lower link reliability in directional RTS/CTS networks.


\begin{figure}
    \centering
    \includegraphics[width=0.45\textwidth]{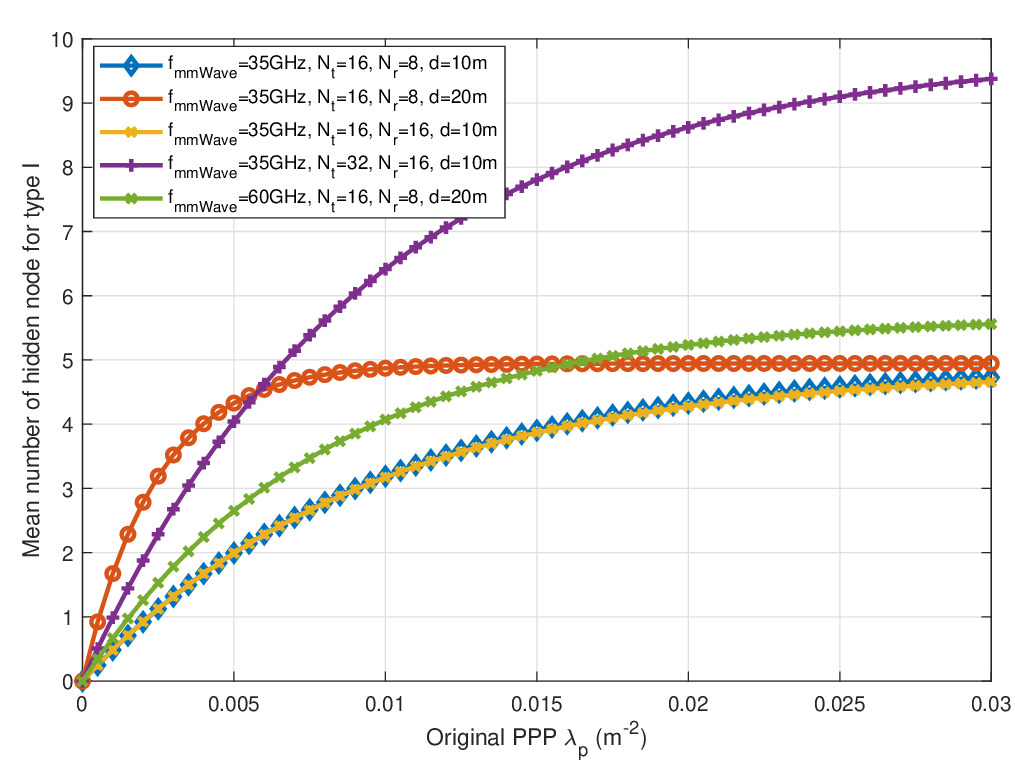}
    \caption{The mean number of hidden node for Type \Rmnum{2} under different system configurations}
    \label{fig:number2}
\end{figure}

\section{Conclusion}
\label{sec:conclusion}
This paper presented a marked point process framework to characterize the spatial configuration of transceivers operating under the RTS/CTS handshake mechanism in WLANs. The proposed model accommodates both omnidirectional and directional communication scenarios, enabling a unified analysis of interference and throughput performance.

Numerical results demonstrate distinct performance trends between the cross-link RTS/CTS and the directional RTS/CTS schemes. The cross-link RTS/CTS ensures more reliable link quality by effectively mitigating interference, but at the cost of reduced overall throughput. In contrast, the directional RTS/CTS achieves higher network throughput, yet suffers from degraded link reliability and a pronounced hidden node problem. 
These findings reveal a fundamental design trade-off between link reliability and spatial reuse. The appropriate choice between the two protocols should thus depend on the network’s target performance objectives—whether prioritizing throughput efficiency or communication robustness.

\appendices

\section{Proof of Lemma 1}
\label{appendix:D} 
The asymptotic gain $G$ is defined as 
\begin{align}
\label{G}
    G=\frac{\mathrm{MISR}_{R}}{\mathrm{MISR}_{\rm G-HCP}},
\end{align}
where $\mathrm{MISR}_{\rm G-HCP}$ is the mean interference-to-average-signal ratio of the G-HCP model, denoted as
\begin{align}
    \mathrm{MISR}_{\rm G-HCP}=\frac{\mathbb{E}_{(o, 0)}^{!}[I_{x_{o},R}]}{l(d)}.
\end{align}

Next, we focus on $\mathrm{MISR}_{R}$, which is the mean interference-to-average-signal ratio of the PPP model with nearest transmitter association and a LOS radium R.
Letting $R_k$ be the distance from the typical user to the $k$-th nearest cellular BS, $v_k= R_1/R_k$ is the distance ratio.
The interference-to-average-signal ratio for the LOS region, denoted by $\mathrm{I\bar{S}R}_{R}$, is defined as
\begin{align}
     \mathrm{{I\bar{S}R}}_{R}=\sum_{i=2}^{\infty}{v_{i}^{\alpha}}{\mathbb{P}(N \geq i)},
\end{align}
where ${\rm \mathbb{P}}(N \geq i)$ denotes the probability that the number of nodes is no less than $i$ within the LOS region.
For a PPP where the mean number of nodes in the LOS region is $N_{\rm LR} = \lambda_p \pi R^2$, ${\mathbb{P}(N \geq i)}$ follows as
\begin{align}
\label{pPPP}
    \mathbb{P}(N \geq i)=\sum_{k=i}^{\infty}\frac{N_{\rm LR}^ke^{-N_{\rm LR}}}{i!}.
\end{align}
Therefore, $\mathrm{ MISR}(R)$ can be written as
\begin{align}
\label{MISRex}
   \mathrm{ MISR}_{R}&\triangleq \nonumber
    \mathbb{E}[\mathrm{{I\bar{S}R}}_{R}]\\
    &=e^{-N_{\rm LR}}\sum_{i=2}^{\infty}{\frac{\Gamma(1+\frac{\alpha}{2}) \Gamma(i)}{\Gamma(i+\frac{\alpha}{2})}} \sum_{k=i}^{\infty}\frac{N_{\rm LR}^k}{k!}.
\end{align}
And $ \mathrm{ MISR}_{\infty}$ follows as \cite{6897962}
\begin{align}
\nonumber
    \mathrm{ MISR}_{\infty}\triangleq\mathbb{E}(\mathrm{I\bar{S}R}_{\infty})=\frac{2}{\alpha-2}.
\end{align}

Since the direct evaluation of (\ref{MISRex}) is computationally challenging, we introduce $g(R)$ defined as:
\begin{align}
\label{f}
     g(R)&= \nonumber
      \mathrm{ MISR}_{\infty}-\mathrm{MISR}_{R}\\
    &=e^{-N_{\rm LR}}\sum_{i=1}^{\infty}{\frac{i!}{(1+\frac{\alpha}{2})_{i}}} \sum_{k=0}^{i}\frac{N_{\rm LR}^k}{k!},
\end{align}
where $(x)_{i}=\frac{\Gamma(x+i)}{\Gamma(x)}$ is the Pochhammer symbol.
Changing the order of summation, (\ref{f}) can be rewritten as
\begin{align} 
    g(R)&= \nonumber
    e^{-N_{\rm LR}}\left(\sum_{k=0}^{\infty}\frac{N_{\rm LR}^k}{k!}\sum_{i=k}^{\infty}{\frac{i!}{(1+\frac{\alpha}{2})_{i}}}-1 \right)\\
    \nonumber
    &=e^{-N_{\rm LR}}\left(\sum_{k=0}^{\infty}\frac{N_{\rm LR}^k}{(1+\frac{\alpha}{2})_{k}}\sum_{i=0}^{\infty}{\frac{(k+1)_{i}}{(k+1+\frac{\alpha}{2})_{i}}}-1 \right)
    \\
    \nonumber
    &\overset{(b)}{=}e^{-N_{\rm LR}}\left(\sum_{k=0}^{\infty}\frac{N_{\rm LR}^k(2k+\alpha)}{(1+\frac{\alpha}{2})_{k}(\alpha-2)}-1 \right)
    \\
    &\overset{(c)}{=}\frac{\alpha}{\alpha-2}F_1(N_{\rm LR})+\frac{4N_{\rm LR}}{\alpha^{2}-4}F_{2}(N_{\rm LR})-e^{-N_{\rm LR}}.
\end{align}
where step (b) follows from the Gauss's Hypergeometric Theorem and (c) follows from the Kummer's transformation for the confluent hypergeometric function with $F_n(x)$ denoting
    \begin{align}
        F_n(x)={}_{1}F_{1}\bigg(\frac{\alpha}{2};n+\frac{\alpha}{2};-x\bigg).
    \end{align}



Therefore, $\mathrm{ MISR}_R$ can be written as
\begin{align}
\nonumber
   \mathrm{ MISR}_R&=\frac{2}{\alpha-2}-g(R),
\end{align}
which can be substituted into (\ref{G}) to get the equation in Lemma 1.
\section{Proof of Theorem 2}
\label{appendix:B}
The cumulative distribution (cdf) of the SIR is
\begin{align}
    F_{\rm SIR}(\theta)=\mathbb{P}(h<\theta \mathrm{I\bar{S}R}),
    \end{align}
    where $h$ follows a gamma distribution Gamma$(M, \frac{1}{M})$.

    Considering that $M\in\mathbb{N}$ and $x\to 0$, the cdf of $h$ is
    \begin{align}
    \nonumber
     F_{h}(x) &= 1-\sum_{k=0}^{M-1} \frac{(Mx)^k e^{-Mx}}{k!} \\
     &\sim \frac{M^{M-1}}{\Gamma(M)} x^M.
\end{align}
    Thus, we have 
    \begin{align}
    \nonumber
    \mathbb{P}(h<\theta {\rm I\bar{S}R} \mid {\rm I\bar{S}R} )
    \nonumber
     &=  F_{h}(\theta {\rm I\bar{S}R})  \\
    &\sim \frac{M^{M-1}}{\Gamma(M)} \theta^M \cdot \mathrm{I\bar{S}R}^M.
    \end{align}
    Taking the expectation, we obtain
    \begin{align}
        \mathbb{P}(\mathrm{SIR}<\theta) \sim  \frac{M^{M-1}}{\Gamma(M)} \theta^M \cdot \mathbb{E}({\rm I\bar{S}R}^{M}).
        \label{P}
    \end{align}

    Based on the traditional ASAPPP, the approximation of $\mathbb{P}(h>\theta \mathrm{I\bar{S}R})$ is written as 
    \begin{align}
    \label{appp}
        \mathbb{P}(h>\theta \mathrm{I\bar{S}R}) \approx \mathbb{P}_{\rm PPP}(h>\theta /G_{M}),
    \end{align}
    where 
     \begin{align}
        G_{M} = \left( \frac{\mathbb{E}(\mathrm{I\bar{S}R}^{M})}{\mathbb{E}(\mathrm{I\bar{S}R}_{\rm G-HCP}^{M})} \right)^{\frac{1}{M}}.
    \end{align}
  The approximation of $G_{M}$ is given by \cite{6897962}
   \begin{align}
   \label{appG}
        G_{M} \approx G,
    \end{align}
    where $G$ is given by Lemma 1.
    
Under Nakagami-$M$ fading conditions, the success probability $\mathbb{P}_{\rm PPP}(\mathrm{SIR}>\theta )$ in the cellular network for a given SIR threshold \(\theta\) is given in \cite{7913628}.
Substituting (\ref{appG}) into (\ref{appp}), Theorem 2 has been proved.

\section{Proof of Theorem 3}
\label{appendix:C}
First, we present the reduced second-order factorial measure $\mathcal{K}_{(r,\beta,\theta)}(B\times L)$ for the marked point process, under the assumption that the transmitter corresponding to $(r,\beta,\theta)$ is contained within the point process $\Phi$.
And $\lambda_b \mathcal{K}_{(r,\beta,\theta)}(B\times L)$ is the expected number of points located in B with marks taking values in L under the condition of the potential transmitter $y=(r,\beta)\in {\Phi}$.

Based on $\mathcal{K}_{(r,\beta,\theta)}(B\times L), B\subset\mathbb{R}^{2}, L\subset[0,2\pi]$, the expectation of hidden nodes $\mathbb{E}_{(o,0)}^![N]$ can be written as
\begin{equation}\begin{aligned}
 \label{hiddennode}
\mathbb{E}_{(o, 0)}^!\left[N_{\rm h}\right] &=\mathbb{E}_{(o,0)}^{!}\left(\sum_{(y,\theta)\in\tilde{\Phi}}\mathds{1}(x_o \in S_{y})\right)\\
 & =\lambda_b\int_{\mathbb{R}^{2}\times[0,2\pi]}\mathds{1}(x_o \in S_{\rm t}(y))\mathcal{K}_{(o, 0)}(\mathrm{d}(y,\theta)).
\end{aligned}\end{equation}

The second-order factorial moment measure for the marked point process can be formally expressed as
 \begin{equation}\begin{aligned}
 \label{3}
 & \alpha^{(2)}(B_{1}\times B_{2}\times L_{1}\times L_{2}) \\
 & =\mathbb{E}\left(\sum_{y_{1},y_{2}\in{\Phi}}^{\neq}1_{B_{1}}(y_{1})1_{B_{2}}(y_{2})1_{L_{1}}(\theta_{1})1_{L_{2}}(\theta_{2})\right)\\
 & =\mathbb{E}_{y_{1}\in {\Phi}}\left(\mathbb{E}\left(\sum_{y_{2}\in {\Phi}}^{y_{2}\neq y_{1}}1_{B_{2}}(y_{2})1_{L_{2}}(\theta_{2})\right)\right)\\
 & \overset{(a)}{=}\frac{\lambda_b^{2}}{2\pi}\int_{B_{1}\times L_{1}}\mathcal{K}_{(y,\theta)}((B_{2}-y)\times L_{2})\mathrm{d}(y,\theta)\\
 &\overset{(b)}{=}\int_{B_{1}\times L_{1}}\left(\int_{(B_{2}-y_{1})\times L_{2}}\varrho^{(2)}(y_{2},\theta_{1},\theta_{2})\mathrm{d}(y_{2},\theta_{2})\right)\mathrm{d}(y_{1},\theta_{1}),
\end{aligned}
\end{equation}
where (a) follows that 
$\mathbb{E}\left(\sum_{y_{2}\in {\Phi}}^{y_{2}\neq y_{1}}1_{B_{2}}(y_{2})1_{L_{2}}(\theta_{2})\right)$ is the expected number of points located in $B_{2}$ with marks taking values in $L_{2}$ under $y_1\in {\Phi}$, i.e., $\lambda_b \mathcal{K}_{(y_1,\theta)}((B_{2}-y_1)\times L_2)$ and (b) follows the definition of $\alpha^{(2)}$ and the second-order product $\varrho^{(2)}$.

The second-order product density $\varrho^{(2)}$ and $\mathcal{K}_{(o,0)}(B\times L)$ satisfy:
\begin{equation}\mathcal{K}_{(o,0)}(d(y,\theta)) =  \frac{2\pi}{\lambda_b^{2}}\varrho^{(2)}(y,0,\theta)\mathrm{d}(y,\theta),
\end{equation}
where $\mathrm{d}(y,\theta)$ denotes the differential measure in polar coordinates.
Therefore, $\mathbb{E}_{(o,0)}^![N_{\rm h}]$ can be written as
\begin{equation}
\begin{aligned}
\label{varN}
 \mathbb{E}_{(o,0)}^![N_{\rm h}]=\frac{2\pi}{\lambda_b}\int_{\mathbb{R}^{2}\times[0,2\pi]}
 \mathds{1}(x_o \in S_{\rm t}(y)\varrho^{(2)}(y,0,\theta)\mathrm{d}(y,\theta).
\end{aligned}
\end{equation}

$\varrho^{(2)}(y,\theta_{0},\theta)$ admits the explicit form
\begin{equation}
\label{var}
\varrho^{(2)}(y,\theta_{0},\theta)=\frac{\lambda_p^{2}}{4\pi^{2}}k(y,\theta_{0},\theta),
\end{equation}
where $k(y,\theta_{0},\theta)$ represents the probability that two transceiver pairs $(o, \theta_{0})$ and $(y, \theta)$ are both within $\tilde{\Phi}$.

Obviously, under the condition that one transceiver pair is the typical pair, $k(y,\theta_{0},\theta)$ can be written as:
\begin{equation}\begin{aligned}
k(y,\theta_{0},\theta)=k(r,\beta,\theta),
\end{aligned}\end{equation}
where $(r,\beta,\theta)$ is the parameter of the other transceiver pair.

Substituting (\ref{var}) into (\ref{varN}), we get
\begin{equation}\begin{aligned}
\label{en}
\mathbb{E}_{(o,0)}^!\left[N_{\rm h}\right]&=\frac{\lambda_p^{2}}{2\pi\lambda_b}\int_{0}^{\infty}\int_{0}^{2\pi}\int_{0}^{2\pi}\\
 &\mathds{1}(x_o \in S_{\rm t}(y)) k(r,\beta,\theta)r\mathrm{d}r\mathrm{d}\beta\mathrm{d}\theta,
\end{aligned}\end{equation}
where $S_{\rm t}(y)$ is the transmitter's exclusion region of the transceiver pair $(r,\beta,\theta)$.

The key point is to derive $\mathds{1}(x_o \in S_{\rm t}(y))$.
$y$ is the transmitter of the transceiver $(r, \beta, \theta)$.
Assuming that $r_y$ is the distance of $x_o$ and $y$ and $\theta_y$ is the angle difference between the vector $\overrightarrow{yx_o}$ and $\overrightarrow{v}=(\cos(\theta), \sin(\theta))$, if $x_o \in S_{\rm t}(y)$, $r_y$ and $\theta_y$ satisfy:
\begin{equation}\begin{aligned}
\label{rx-thetax}
&|\theta_y| \leq \frac{2\lambda}{\pi N_{\rm t} d_0}\arccos(\frac{r_y}{R_{\rm t}}).
\end{aligned}
\end{equation}

Substituting (\ref{rx-thetax}) into (\ref{en}), (\ref{en}) can be written as:
\begin{equation}\begin{aligned}
\label{en1}
  \mathbb{E}_{(o, 0)}^{!}\left[N_{\rm h}\right]&=\frac{\lambda_p^{2}}{\pi\lambda_b}\int_{0}^{\pi}\int_{0}^{\infty}\int_{-\frac{2\lambda}{\pi N_{\rm t} d_0}\arccos(\frac{r_y}{R_{\rm t}})}^{\frac{2\lambda}{\pi N_{\rm t} d_0}\arccos(\frac{r_y}{R_{\rm t}})}\\
 & k\bigg(r,\beta,\theta-\arccos\Big(\frac{r\cos\beta-d}{r_y}\Big)\bigg)r\mathrm{d}\beta\mathrm{d}r\mathrm{d}\theta,
\end{aligned}\end{equation}
where $r_y=\sqrt{r^2+d^2-2rd\cos\beta}$.

Noting $r_y\leq R_{\rm t}$, $r$ should satisfy:
\begin{equation}\begin{aligned}
\label{maxr}
r\leq R_{\rm t}\sqrt{1-(\frac{d\sin\beta}{R_{\rm t}})^2}-d\cos\beta,
\end{aligned}\end{equation}
which follows the equation:
\begin{equation}\begin{aligned}
r^2+d^2-r_y^2=2rd\cos\beta.
\end{aligned}\end{equation}
Substituting (\ref{maxr}) into (\ref{en1}) has been mentioned above completes the proof of Theorem \ref{hnthm}.
\bibliographystyle{IEEEtran}
\bibliography{csma-isit}

@ARTICLE{1092767,
  author={Tobagi, F. and Kleinrock, L.},
  journal={IEEE Transactions on Communications}, 
  title={{Packet Switching in Radio Channels: Part II - The Hidden Terminal Problem in Carrier Sense Multiple-Access and the Busy-Tone Solution}}, 
  year={1975},
  volume={23},
  number={12},
  pages={1417-1433},
  doi={10.1109/TCOM.1975.1092767}}

@ARTICLE{815305,
  author={},
  journal={IEEE Std 802.11a-1999}, 
  title={{IEEE Standard for Telecommunications and Information Exchange Between Systems - LAN/MAN Specific Requirements - Part 11: Wireless Medium Access Control (MAC) and Physical Layer (PHY) Specifications: High Speed Physical Layer in the 5 GHz band}}, 
  year={1999},
  volume={},
  number={},
  pages={1-102},
  doi={10.1109/IEEESTD.1999.90606}}

@ARTICLE{9793689,
  author={Lindroos, Saku and Hakkala, Antti and Virtanen, Seppo},
  journal={IEEE Access}, 
  title={{Battle of the Bands: A Long-Term Analysis of Frequency Band and Channel Distribution Development in WLANs}}, 
  year={2022},
  volume={10},
  number={},
  pages={61463-61471},
  doi={10.1109/ACCESS.2022.3182011}}

@ARTICLE{620533,
  author={Crow, B.P. and Widjaja, I. and Kim, J.G. and Sakai, P.T.},
  journal={IEEE Communications Magazine}, 
  title={{IEEE 802.11 Wireless Local Area Networks}}, 
  year={1997},
  volume={35},
  number={9},
  pages={116-126},
  doi={10.1109/35.620533}}

@techreport{ericsson2023mobility,
    author      ={Ericsson},
    title       = {{Ericsson Mobility Report November 2023}}, 
    url         = {https://www.ericsson.com/en/reports-and-papers/mobility-report/key-figures}
}

@ARTICLE{9165719,
  author={Naik, Gaurang and Park, Jung-Min and Ashdown, Jonathan and Lehr, William},
  journal={IEEE Access}, 
  title={{Next Generation Wi-Fi and 5G NR-U in the 6 GHz Bands: Opportunities and Challenges}}, 
  year={2020},
  volume={8},
  number={},
  pages={153027-153056},
  doi={10.1109/ACCESS.2020.3016036}}

@ARTICLE{7306370,
  author={Misra, Satyajayant},
  journal={IEEE Wireless Communications}, 
  title={{Millimeter Wave Wireless Communications (Rappaport, T., et al; 2014) [Book review]}}, 
  year={2015},
  volume={22},
  number={5},
  pages={6-7},
  doi={10.1109/MWC.2015.7306370}}

@ARTICLE{7913628,
  author={Yu, Xianghao and Zhang, Jun and Haenggi, Martin and Letaief, Khaled B.},
  journal={IEEE Journal on Selected Areas in Communications}, 
  title={{Coverage Analysis for Millimeter Wave Networks: The Impact of Directional Antenna Arrays}}, 
  year={2017},
  volume={35},
  number={7},
  pages={1498-1512},
  doi={10.1109/JSAC.2017.2699098}}

@ARTICLE{6834753,
  author={Akdeniz, Mustafa Riza and Liu, Yuanpeng and Samimi, Mathew K. and Sun, Shu and Rangan, Sundeep and Rappaport, Theodore S. and Erkip, Elza},
  journal={IEEE Journal on Selected Areas in Communications}, 
  title={{Millimeter Wave Channel Modeling and Cellular Capacity Evaluation}}, 
  year={2014},
  volume={32},
  number={6},
  pages={1164-1179},
  doi={10.1109/JSAC.2014.2328154}}

@ARTICLE{6897962,
  author={Haenggi, Martin},
  journal={IEEE Wireless Communications Letters}, 
  title={{The Mean Interference-to-Signal Ratio and Its Key Role in Cellular and Amorphous Networks}}, 
  year={2014},
  volume={3},
  number={6},
  pages={597-600}
}

@ARTICLE{7322270,
  author={Ganti, Radha Krishna and Haenggi, Martin},
  journal={IEEE Transactions on Wireless Communications}, 
  title={{Asymptotics and Approximation of the SIR Distribution in General Cellular Networks}}, 
  year={2016},
  volume={15},
  number={3},
  pages={2130-2143}
}

@ARTICLE{8648502,
  author={Kalamkar, Sanket S. and Haenggi, Martin},
  journal={IEEE Transactions on Communications}, 
  title={{Simple Approximations of the SIR Meta Distribution in General Cellular Networks}}, 
  year={2019},
  volume={67},
  number={6},
  pages={4393-4406}
}

@ARTICLE{9855457,
  author={Chen, Cheng and Chen, Xiaogang and Das, Dibakar and Akhmetov, Dmitry and Cordeiro, Carlos},
  journal={IEEE Communications Standards Magazine}, 
  title={{Overview and Performance Evaluation of Wi-Fi 7}}, 
  year={2022},
  volume={6},
  number={2},
  pages={12-18},
  doi={10.1109/MCOMSTD.0001.2100082}}

@misc{IEEE802.11_IMMW_2024,
  author       = {{IEEE 802.11 Working Group}},
  title        = {{IMMW System Reuses}},
  howpublished = {PowerPoint presentation},
  month        = mar,
  year         = {2024},
  note         = {IEEE 802.11 Working Group Document 11-25-0238-00-00BQ},
  url          = {https://mentor.ieee.org/802.11/dcn/25/11-25-0238-00-00bq-immw-system-reuses.pptx},
  urldate      = {2024-03-01}
}

@techreport{par,
  author       = {{IEEE P802.11 Wireless LANs}},
  title        = {{802.11 IMMW Proposed PAR}},
  institution  = {IEEE},
  type         = {Standard},
  number       = {802.11-24/0116r7},
  year         = {2024},
  month        = {Jan},
  url          = {https://mentor.ieee.org/802.11/dcn/24/11-24-0116-07-immw-immw-draft-proposed-par.docx},
  note         = {Revision 7}
}

@ARTICLE{10058126,
  author={},
  journal={IEEE P802.11be/D3.0, January 2023}, 
  title={{IEEE Draft Standard for Information technology--Telecommunications and information exchange between systems Local and metropolitan area networks--Specific requirements - Part 11: Wireless LAN Medium Access Control (MAC) and Physical Layer (PHY) Specifications Amendment: Enhancements for Extremely High Throughput (EHT)}}, 
  year={2023},
  volume={},
  number={},
  pages={1-999},
  doi={}}

@INPROCEEDINGS{10978791,
  author={Lian, Shumin and Tong, Jingwen and Fu, Liqun},
  booktitle={2025 IEEE Wireless Communications and Networking Conference (WCNC)}, 
  title={{Dynamic Channel Allocation via Bandit Learning for WiFi 7 Networks with Multi-Link Operation}}, 
  year={2025},
  volume={},
  number={},
  pages={1-6},
  doi={10.1109/WCNC61545.2025.10978791}
}

@ARTICLE{68340753,
  author={Akdeniz, Mustafa Riza and Liu, Yuanpeng and Samimi, Mathew K. and Sun, Shu and Rangan, Sundeep and Rappaport, Theodore S. and Erkip, Elza},
  journal={IEEE Journal on Selected Areas in Communications}, 
  title={{Millimeter Wave Channel Modeling and Cellular Capacity Evaluation}}, 
  year={2014},
  volume={32},
  number={6},
  pages={1164-1179},
  doi={10.1109/JSAC.2014.2328154}
}

@ARTICLE{6932503,
  author={Bai, Tianyang and Heath, Robert W.},
  journal={IEEE Transactions on Wireless Communications}, 
  title={{Coverage and Rate Analysis for Millimeter-Wave Cellular Networks}}, 
  year={2015},
  volume={14},
  number={2},
  pages={1100-1114},
  doi={10.1109/TWC.2014.2364267}}

@ARTICLE{8550813,
  author={Azimi-Abarghouyi, Seyed Mohammad and Makki, Behrooz and Nasiri-Kenari, Masoumeh and Svensson, Tommy},
  journal={IEEE Transactions on Vehicular Technology}, 
  title={{Stochastic Geometry Modeling and Analysis of Finite Millimeter Wave Wireless Networks}}, 
  year={2019},
  volume={68},
  number={2},
  pages={1378-1393},
  doi={10.1109/TVT.2018.2883891}}

@ARTICLE{9494282,
  author={Wei, Haichao and Deng, Na and Haenggi, Martin},
  journal={IEEE Transactions on Wireless Communications}, 
  title={{Performance Analysis of Inter-Cell Interference Coordination in mm-Wave Cellular Networks}}, 
  year={2022},
  volume={21},
  number={2},
  pages={726-738},
  doi={10.1109/TWC.2021.3097376}}

@INPROCEEDINGS{6736994,
  author={Bai, Tianyang and Heath, Robert W.},
  booktitle={2013 IEEE Global Conference on Signal and Information Processing}, 
  title={{Coverage analysis for millimeter wave cellular networks with blockage effects}}, 
  year={2013},
  volume={},
  number={},
  pages={727-730},
  doi={10.1109/GlobalSIP.2013.6736994}}

@ARTICLE{chen2024characterizing,
  author={Chen, Zhuoling and Zhong, Yi},
  journal={IEEE Wireless Communications Letters}, 
  title={{Characterizing Stability Regions in Wireless Networks With Hard-Core Spatial Constraints}}, 
  year={2024},
  volume={13},
  number={8},
  pages={2180-2184}}

@ARTICLE{5934671,
  author={Haenggi, Martin},
  journal={IEEE Communications Letters}, 
  title={{Mean Interference in Hard-Core Wireless Networks}}, 
  year={2011},
  volume={15},
  number={8},
  pages={792-794},
  doi={10.1109/LCOMM.2011.061611.110960}}

@ARTICLE{11018845,
  author={Zhong, Yi and Chen, Zhuoling and Zhang, Wenyi and Haenggi, Martin},
  journal={IEEE Transactions on Wireless Communications}, 
  title={{Dual-Zone Hard-Core Model for RTS/CTS Handshake Analysis in WLANs}}, 
  year={2025},
  volume={},
  number={},
  pages={Early access},
  doi={10.1109/TWC.2025.3572687}}

@INPROCEEDINGS{11123194,
  author={Chen, Zhuoling and Zhong, Yi and Yang, Howard H.},
  booktitle={23rd International Symposium on Modeling and Optimization in Mobile, Ad Hoc, and Wireless Networks (WiOpt)}, 
  title={{Interference in Millimeter-Wave Systems with Directional RTS/CTS Handshake}}, 
  year={2025},
  volume={},
  number={},
  pages={1-7},
  doi={10.23919/WiOpt66569.2025.11123194}}

@ARTICLE{11192202,
  author={Zhong, Yi and Zhou, Xiaohang and Feng, Ke},
  journal={IEEE Transactions on Wireless Communications}, 
  title={{Spatial Network Calculus: Toward Deterministic Wireless Networking}}, 
  year={2025},
  volume={},
  number={},
  pages={Early access},
  doi={10.1109/TWC.2025.3613451}}

@INPROCEEDINGS{4215725,
  author={Nguyen, H. Q. and Baccelli, F. and Kofman, D.},
  booktitle={IEEE INFOCOM 2007 - 26th IEEE International Conference on Computer Communications}, 
  title={{A Stochastic Geometry Analysis of Dense IEEE 802.11 Networks}}, 
  year={2007},
  volume={},
  number={},
  pages={1199-1207},
  doi={10.1109/INFCOM.2007.143}}

\end{document}